\documentclass[11pt]{article}
\renewcommand{\deg}{\mathsf{deg}}

\usepackage{cite}
\usepackage{amsmath}
\usepackage{appendix}
\usepackage{graphicx}
\usepackage{color}
\usepackage{algorithm}
\usepackage[noend]{algpseudocode}
\usepackage{epstopdf}
\usepackage{wrapfig}
\usepackage[textsize=tiny]{todonotes}

\def\SETu{S_{\mbox{\sf\scriptsize uni}}}
\def\SETt{S_{\mbox{\sf\scriptsize tmp}}}
\def\NNu{\NN_{\mbox{\sf\scriptsize uni}}}
\def\NNt{\NN_{\mbox{\sf\scriptsize tmp}}}
\def\GAMm{\NN_{\mbox{\sf\scriptsize many}}}
\def\NN{N}
\def\EN{\Gamma^{-}(S)}
\def\EEu{E_{\mbox{\sf\scriptsize uni}}}
\def\EEt{E_{\mbox{\sf\scriptsize tmp}}}
\def\UN{\Gamma^1}

\def\partition{\mbox{\tt Partition}}
\def\gain{\mbox{\tt gain}}

\long\def\commabs #1\commabsend{}
\long\def\commful #1\commfulend{#1}
\long\def\comment #1\commentend{}

\newcommand {\ignore} [1] {}
\usepackage[framemethod=tikz]{mdframed}
\usepackage{caption}
\usepackage{xspace}
\def\inline#1:{\par\vskip 7pt\noindent{\bf #1:}\hskip 10pt}

\allowdisplaybreaks
\def\Invariant{\par\noindent{\bf Invariant:~}}
\def\blackslug{\hbox{\hskip 1pt \vrule width 4pt height 8pt
    depth 1.5pt \hskip 1pt}}

\def\denseformat{
\setlength{\textheight}{9.5in}
\setlength{\textwidth}{6.9in}
\setlength{\evensidemargin}{-0.3in}
\setlength{\oddsidemargin}{-0.3in}
\setlength{\headsep}{10pt}
\setlength{\topmargin}{-0.44in}
\setlength{\columnsep}{0.375in}
\setlength{\itemsep}{0pt}
}
\newtheorem{theorem}{Theorem}[section]

\newtheorem{claim}[theorem]{Claim}
\newtheorem{lemma}[theorem]{Lemma}

\newtheorem{corollary}[theorem]{Corollary}

\newtheorem{observation}[theorem]{Observation}

\def\boldhead#1:{\par\vskip 7pt\noindent{\bf #1:}\hskip 10pt}
\def\ithead#1:{\par\vskip 7pt\noindent{\it #1:}\hskip 10pt}

\def\inline#1:{\par\vskip 7pt\noindent{\bf #1:}\hskip 10pt}
\def\midinline#1:{\par\noindent{\bf #1:}\hskip 10pt}
\def\dnsinline#1:{\par\vskip -7pt\noindent{\bf #1:}\hskip 10pt}
\def\ddnsinline#1:{\newline{\bf #1:}\hskip 10pt}
\def\largeinline#1:{\par\vskip 7pt\noindent{\large\bf #1:}\hskip 10pt}
\long\def\commhide #1\commhideend{}
\long\def\commtim #1\commendt{#1}
\long\def\commb #1\commbend{}
\long\def\commedit #1\commeditend{} % Editing comments, marked also by $>>>$

\long\def\commB #1\commBend{}       % Omit in 1996 (both TR and Siena)
\long\def\commex #1\commexend{}     % LN home exercise (hide solutions)

\long\def\commsiena #1\commsienaend{}  % omit in Siena, show in TR

\long\def\commBI #1\commBIend{}  % omit in Bar-Ilan

\long\def\CProof #1\CQED{}

\def\blackslug{\hbox{\hskip 1pt \vrule width 4pt height 8pt
    depth 1.5pt \hskip 1pt}}
\def\QED{\quad\blackslug\lower 8.5pt\null\par}
\def\inQED{\quad\quad\blackslug}

\def\Proof{\par\noindent{\bf Proof:~}}

\long\def\PPP#1{\noindent{\bf Proof:}{ #1}{\quad\blackslug\lower 8.5pt\null}}

\long\def\denspar #1\densend
{#1}
\newif\ifnotesw\noteswtrue% T to show box & marginal notes; F supresses.
   {\ifnotesw\marginpar[\hfill\(\top\)]{\(\top\)}\fi}%
   {\ifnotesw\marginpar[\hfill\(\bot\)]{\(\bot\)}\fi}

\newcommand{\mnote}[1]%
    {\ifnotesw\marginpar%
        [{\scriptsize\it\begin{minipage}[t]{\marginparwidth}
        \raggedleft#1%
                        \end{minipage}}]%
        {\scriptsize\it\begin{minipage}[t]{\marginparwidth}
        \raggedright#1%
                        \end{minipage}}%
    \fi}

\def\MathF{\hbox{\rm I\kern-2pt F}}
\def\MathP{\hbox{\rm I\kern-2pt P}}
\def\MathR{\hbox{\rm I\kern-2pt R}}
\def\MathZ{\hbox{\sf Z\kern-4pt Z}}
\def\MathN{\hbox{\rm I\kern-2pt I\kern-3.1pt N}}
\def\MathC{\hbox{\rm \kern0.7pt\raise0.8pt\hbox{\footnotesize I}
\kern-4.2pt C}}
\def\MathQ{\hbox{\rm I\kern-6pt Q}}

\def\MathE{\hbox{{\rm I}\hskip -2pt {\rm E}}}

\newsavebox{\ttop}\newsavebox{\bbot}

\newcommand{\Exp}[1]{\MathE\left( #1 \right)}
\newcommand{\Prob}{\MathP}

\denseformat

\begin{document}

\title{Wireless Expanders}

\author{Shirel Attali\thanks{The Weizmann Institute of Science. Email: {\tt shirel.attali@weizmann.ac.il}.}
\and
Merav Parter\thanks{The Weizmann Institute of Science. Email: {\tt merav.parter@weizmann.ac.il}.}
\and
David Peleg\thanks{The Weizmann Institute of Science. Email: {\tt david.peleg@weizmann.ac.il}.}
\and
Shay Solomon\thanks{IBM Research. Email: {\tt solo.shay@gmail.com}.}}

\date{\empty}

\begin{titlepage}
\def\thepage{}
\maketitle

%%%%%%%%%%%%%%%%%%%%%%

\maketitle

\begin{abstract}
This paper introduces an extended notion of expansion suitable for
radio networks. A graph $G=(V,E)$ is said to be an
$(\alpha_w, \beta_w)$-{\em wireless expander} if for every subset
$S \subseteq V$ s.t. $|S|\leq \alpha_w \cdot |V|$, there exists a subset
$S'\subseteq S$  s.t. there are at least $\beta_w \cdot |S|$ vertices in
$V\backslash S$ that are adjacent in $G$ to exactly one vertex in $S'$.
The main question we ask is the following: to what extent are ordinary expanders also good \emph{wireless} expanders? We answer this question in a nearly tight manner. On the positive side, we show that any $(\alpha, \beta)$-expander with maximum degree $\Delta$ and $\beta\geq 1/\Delta$ is also a  $(\alpha_w, \beta_w)$ wireless expander for $\beta_w =  \Omega(\beta / \log (2 \cdot \min\{\Delta / \beta, \Delta \cdot \beta\}))$.
Thus the wireless expansion can be smaller than the ordinary expansion by at most a factor that is logarithmic in
$\min\{\Delta / \beta, \Delta \cdot \beta\}$, which, in turn, depends on the \emph{average degree} rather than the maximum degree of the graph.
In particular, for low arboricity graphs (such as planar graphs), the wireless expansion matches the ordinary expansion up to a constant factor.
We complement this positive result by presenting an explicit construction of a ``bad'' $(\alpha, \beta)$-expander for which the wireless expansion is $\beta_w = O(\beta / \log (2 \cdot \min\{\Delta / \beta, \Delta \cdot \beta\})$.

We also analyze the theoretical properties of wireless expanders and their connection to unique neighbor expanders, and then demonstrate their applicability:
Our results (both the positive and the negative) yield improved bounds for the \emph{spokesmen election problem} that was introduced in the seminal paper of Chlamtac and Weinstein \cite{chlamtac1991wave} to devise efficient broadcasting for multihop radio networks.
Our negative result yields a significantly simpler proof than that from the seminal paper of Kushilevitz and Mansour \cite{kushilevitz1998omega}
for a lower bound on the broadcast time in radio networks.
\end{abstract}
\end{titlepage}

\pagenumbering{arabic}
%\newpage
%%%%%%%%%%%%%%%%%%%%%
\section{Introduction}
\label{Introduction}
%%%%%%%%%%%%%%%%
\subsection{Background and motivation}
An expander is a sparse graph that has strong connectivity properties
\cite{hoory2006expander}. There are several definitions for expanders,
with natural connections between them.
We focus on the following combinatorial definition.
\inline Expanders: Let $G=(V,E)$ be an undirected graph.
For a set $S\subset V$, let $\Gamma(S)$ denote the set of neighbors of the verices of $S$, and define $\EN=\Gamma(S)\setminus S$. We say that $G$ is an $(\alpha, \beta)$ {\em expander}, for positive parameters $\alpha$ and $\beta$, if $|\EN|\geq \beta \cdot |S|$ for every $S \subseteq V$ s.t. $|S|\leq \alpha \cdot |V|$.

One of the main advantages of expanders is that they enable fast and effective
dissemination of information from a small group of vertices to the outside
world. This property becomes less immediate when we consider using
%Our main interest in this study is in applying
the expansion property in the context of wireless communication networks.
Such networks can be represented by a specific kind of graphs,
called {\em radio networks} \cite{chlamtac1985broadcasting}.
A radio network is an undirected (multihop) network of processors
that communicate in synchronous rounds in the following manner.
In each step, a processor can either transmit or keep silent.
A processor receives a message in a given step if and only if it keeps silent
and precisely one of its neighbors transmits in this step.
If none of its neighbors transmits, it hears nothing.
If more than one neighbor (including itself) transmits in a given step,
then none of the messages is received.
In this case we say that a \emph {collision} occurred.
It is assumed that the effect at processor $u$ of more than one of its
neighbors transmitting is the same as of no neighbor transmitting,
i.e., a node cannot distinguish a collision from silence.

%As we are going to see next,
The usual definition
of expanders is not enough to ensure fast message propagation in radio networks.
Consider, for example, a radio network $C^+$ consisting of a complete graph $C$
with one more vertex $s_0$, the source, connected to two vertices $x$ and $y$
from $C$. Obviously this is a good expander, but in this case, after the first
step of broadcast, if all the vertices that received the message
(i.e., the three vertices $s_0$, $x$ and $y$) transmit it simultaneously
to all their neighbors, then no one will hear it.
This motivates considering another definition of expanders, namely,
{\em unique neighbor expanders} (or unique expanders, in short)\cite{AC02}.
%One of these tasks is message broadcasting. Broadcast is a process
%by which a message $M$, initiated by a processor $s_0$ (the source)
%is delivered to all other processors in the network.
%In ordinary (point to point message passing) networks, one may expect that
%if the graph is a good expander then broadcast will be performed fast
%(by requiring all the vertices that received the message to send it
%to all their neighbors).
%\vspace{-5pt}
\inline Unique neighbor expanders: Let $G=(V,E)$ be an undirected graph. We say that $G$ is
an $(\alpha_u, \beta_u)$-{\em unique neighbor expander} if for every $S \subseteq V$
s.t. $|S|\leq \alpha_u \cdot |V|$, there are at least $\beta_u \cdot |S|$
vertices in $V\backslash S$ that are adjacent in $G$ to exactly one vertex
in $S$.

Clearly, if $G$ is a unique expander with good parameters,
then broadcasting on it can be fast
(again, by requiring all the vertices that received the message
to send it to all their neighbors).
Unfortunately, it seems that unique neighbor expansion might be hard to come by.
For example, while the graph $C^+$ described above is a good (ordinary) expander, it is clearly not a good unique expander, as can be realized by considering the set $S=\{x,y,s_0\}$. (In general, ordinary expanders might have rather small unique neighbor expansion, as will be shown soon.) In addition, explicit constructions of unique expanders are rather scarce and known only for a limited set of parameters \cite{AC02,becker2016symmetric}.

The key observation triggering the current paper is that the property
required from unique expanders might be \emph{stronger} than necessary.
This is because there is no reason to require {\em all} the vertices
that received the message to send it. Rather, it may be enough to pick
a {\em subset} $X$ of this set, that has a large set of unique neighbors,
and require only the vertices of $X$ to transmit.
This may be an attractive alternative since such a property may be easier
to guarantee than unique neighbor expansion, and therefore may be achievable
with better parameters $\alpha$ and $\beta$.
(Note, e.g., that this property holds for our example graph $C^+$.)
This observation thus motivates our definition for a new variant of expanders.
%\vspace{-5pt}
\inline Wireless expanders: Let $G=(V,E)$ be an undirected graph. We say that
$G$ is an $(\alpha_w, \beta_w)$-{\em wireless expander} if for every
$S \subseteq V$
s.t. $|S|\leq \alpha_w \cdot |V|$, there exists a subset $S'\subseteq S$
s.t. there are at least $\beta_w \cdot |S|$ vertices in $V\backslash S$
that are adjacent to exactly one vertex in $S'$.

In this paper we are interested in investigating the properties of
wireless expanders and the relationships between these graphs
and the classes of ordinary expanders and unique neighbor expanders. We ask the following questions: by how much does the relaxed definition of wireless expanders (compared to unique neighbor expanders) help us in providing expanders with better parameters that are suitable for radio network communication? More specifically, given an $(\alpha, \beta)$-expander, can we prove that it is also
an $(\alpha_w, \beta_w)$-wireless expander with $\alpha_w= f(\alpha, \beta)$
and $\beta_w=g(\alpha, \beta)$, for some functions $f$ and $g$?

%The other direction is to investigate the usefulness of wireless expander
%graphs for various applications.
%In particular, we are interested in finding whether certain applications,
%such as broadcasting in a radio network,
%(see \cite{DBLP:journals/jcss/AlonBLP91},
%\cite{DBLP:journals/jcss/Bar-YehudaGI92},  \cite{CW91}),
%can be simplified or even improved on a wireless expander.

%%%%%%%%%%%%%%%%%%%%%%%%%%%%%%%
\subsection{Our Contribution}

We present several results relating the parameters
of the different notions of expanders.
We begin by investigating the relationships between ordinary expanders and
the more strict notion of unique neighbor expanders.
\begin{itemize}

\item Let $G=(V,E)$ be a $d$-regular graph that is an $(\alpha_u, \beta_u)$-unique neighbor expander, and let $\lambda=\lambda_2$ denote the second largest eigenvalue of its adjacency matrix, given by $a_{uv}=1$ if $(u,v)\in E$ and $a_{uv}=0$ otherwise. Then $G$ is an $(\alpha, \beta)$-expander with $\alpha= \alpha_u$ and $\beta\geq (1-1/d) \cdot\beta_u + (d- \lambda)/d \cdot (1-\alpha_u)$.

\item Suppose $G=(V,E)$ is an $(\alpha,\beta)$-expander with maximum degree $\Delta$. Then
it is also an $(\alpha_u ,\beta_u )$-unique expander with $\alpha_u=\alpha$, and $\beta_u\geq 2\beta-\Delta$. On the other hand, we show that there is an $(\alpha,\beta)$ bipartite expander whose unique expansion is $\beta_u\leq 2\beta-\Delta$.
\end{itemize}
We then turn to consider our new relaxed notion of wireless expander. Our key contribution is in providing nearly tight characterization for the relation between ordinary expanders and wireless expanders. On the positive side,
using the probabilistic method, we show:
\begin{mdframed}[hidealllines=true,backgroundcolor=gray!25]
%\vspace{-8pt}
\begin{theorem}[Positive Result] \label{thm:lb} For every $\Delta\geq 1$, $\beta \geq 1/\Delta$, every $(\alpha,\beta)$-expander $G$ with maximum degree $\Delta$ is a also an ($\alpha_w ,\beta_w$)-wireless expander with $\alpha_w \geq \alpha$ and
$$\beta_w = \Omega(\beta / \log (2\cdot \min\{\Delta / \beta, \Delta \cdot \beta\})).$$
%\vspace{-3pt}
\end{theorem}
\end{mdframed}
Our probabilistic argument has some similarity to the known \emph{decay} method \cite{Bar-YehudaGI87}, which is a standard technique for coping with collisions
in radio networks. Roughly speaking, in the decay protocol of \cite{Bar-YehudaGI87}, time is divided into phases of $\log n$ rounds and in the $i^{th}$ round of each phase, each node that holds a message transmits it with probability $2^{-i}$. Hence, each node that has a neighbor that holds a message, receives it within $O(\log n)$ phases. We use the idea of the decay method to show the existence of a subset $S' \subseteq S$ with a large unique neighborhood in $\Gamma(S)$.

An important feature of our argument is that it bounds the deviation of the wireless expansion from the ordinary expansion as a function of the \emph{average-degree} rather than the maximum degree.
As $\beta$ gets closer to $\Delta$ or to $1/\Delta$,
this finer dependence leads to significantly better results than what could be achieved using the standard decay argument;
our argument
%leads to stronger results
is also arguably simpler than the standard decay argument.
As a technical note, we use the probabilistic method to prove a lower bound of $\Omega(\beta / \log (2\cdot \Delta / \beta))$ on $\beta_w$, and then we push it up to
the bound of Theorem \ref{thm:lb} via a separate \emph{deterministic} argument.
As a corollary, for the important family of low \emph{arboricity graphs}, which includes planar graphs and more generally graphs excluding a fixed minor,
%we have $\min\{\Delta / \beta, \Delta \cdot \beta\} = O(1)$, hence
the wireless expansion matches the ordinary expansion up to a constant factor.
(Indeed, the arboricity is at least $\min\{\Delta / \beta, \Delta \cdot \beta\}$; see Section \ref{Notation and Definitions} for the definition of arboricity.)
In particular, this shows that radio broadcast in low arboricity graphs can be done much more efficiently than what was previously known!

Beyond the probabilistic argument, we also provide explicit deterministic arguments that obtain better parameters (by a constant factor); these are deferred to the appendix.

We also show that asymptotically, no tighter connection can be established:
\begin{mdframed}[hidealllines=true,backgroundcolor=gray!25]
%\vspace{-8pt}
\begin{theorem}[Negative Result] \label{thm:ub} There exists an $(\alpha,\beta)$-expander with maximum degree $\Delta$,
whose wireless expansion is $\beta_w = O(\beta / \log (2\cdot \min\{\Delta / \beta, \Delta \cdot \beta\})$.
\end{theorem}
\end{mdframed}
The explicit construction of this bad graph example is perhaps the most technically challenging result of this paper.
Our explicit construction has interesting connections to related constructions that have been studied in the context of broadcast in radio networks \cite{kushilevitz1998omega,alon2014broadcast}. For instance, our ``core graph'' from Section \ref{core} is reminiscent of a fundamental construction from
\cite{alon2014broadcast}. However, while the construction of \cite{alon2014broadcast} is implicit (using the probabilistic method), our construction is explicit and can be viewed in some sense as the deterministic counterpart of \cite{alon2014broadcast}; moreover, our construction is arguably much simpler than that of \cite{alon2014broadcast}.
We view this explicit construction as the technical highlight of our work, and anticipate that it will find further applications.

An additional application of both our positive and negative results is to the \emph{Spokesman Election problem} introduced in the seminal paper of \cite{chlamtac1991wave}, where given a bipartite graph $G=(S,N,E)$, the goal is to compute a subset $S' \subseteq S$ with the maximum number of unique neighbors $\UN(S')$ in $N$.
More specifically, we provide tight bounds for this problem, which apply to any expansion and average degree parameters,
%i.e., for each choice of these parameters we have an example of a bipartite graph that provides a tight lower bound,
whereas the previous result of \cite{chlamtac1991wave} applies only to one specific (very large) expansion parameter
and only with respect to the maximum degree (rather than the average degree, which is a finer measure).
In Section \ref{rel}, we provide a detailed comparison to the bounds obtained by \cite{chlamtac1991wave}.

Finally, another application of our negative result, and of our explicit core graph in particular, is in the context of broadcast lower bounds in radio networks.
In their seminal paper, Kushilevitz and Mansour \cite{kushilevitz1998omega} proved that there exist networks
in which the expected time to broadcast a message is $\Omega(D \log(n/D))$, where $D$ is the network diameter and $n$ is the number of vertices,
and this lower bound is tight for any $D = \Omega(\log n)$ due to a highly nontrivial upper bound by Czumaj and Rytter \cite{CzumajR06}.
Since the upper bound of \cite{CzumajR06} holds with high probability, it   implies that the lower bound $\Omega(D \log(n/D))$ of \cite{kushilevitz1998omega} also holds
with high probability.
Newport \cite{DBLP:conf/wdag/Newport14} presented an interesting alternative proof to the one by  Kushilevitz and Mansour.
Although short and elegant, Newport's proof relies on two fundamental results in this area, due to Alon et al.\ \cite{DBLP:journals/jcss/AlonBLP91}
and Alon et al.\ \cite{alon2014broadcast} -- Lemma 3.1 in \cite{DBLP:conf/wdag/Newport14} -- whose proof is intricate.
Also, as with Kushilevitz and Mansour's proof, Newport only proves an \emph{expected} lower bound on the broadcast time,
with the understanding that a high probability bound follows from  \cite{CzumajR06}.
By unwinding the ingredients of Newport's proof, the resulting proof (especially for a high probability bound on the broadcast time) is long and intricate.
%On other other hand, our proof is self-contained, it only relies on our simple, explicit core graph construction.
%Moreover, and perhaps more importantly, we provide the **first direct proof for a high probability lower bound** on the broadcast time (in addition to an expected lower bound).
%We thus believe that our proof has important advantages over the previous ones, and we view it as one of our key results (it was a mistake to place part of this proof in the %appendix, we will change this).
%There is also a lower bound of $\Omega(\log^2 n)$ by \cite{DBLP:journals/jcss/AlonBLP91} on the broadcast time.
%which In the complementary regime of diameter $D = O(\log n)$, there is
%The proof of \cite
Using the properties of our explicit core graph construction, we derive a   simple and self-contained proof for the same lower bound, arguably much simpler than that of \cite{kushilevitz1998omega, {DBLP:conf/wdag/Newport14}}.
An important advantage of our proof over
\cite{kushilevitz1998omega, {DBLP:conf/wdag/Newport14}}
is that it gives a high probability bound on the broadcast time \emph{directly},
i.e., without having to take a detour through the upper bound of \cite{CzumajR06}.
%on the other hand, the previous high probability bound proof follows by \emph{combining} the lower bound of \cite{kushilevitz1998omega} (which holds in expectation)
%with the matching upper bound of \cite{CzumajR06} (which holds with high probability), and are much more intricate than ours.

Summarizing, besides the mathematical appeal of wireless expanders and their connections to well-studied types of expanders,
we demonstrate that they find natural applications in the well-studied area of radio networks.
We anticipate that a further study of wireless expanders will reveal additional applications, also outside the scope of radio networks, and we thus
believe it is of fundamental importance.
\subsection{Organization}
In Section\ \ref{sec:prel} we introduce the   notation and definitions used throughout.
We investigate the relations between ordinary expanders and unique neighbor expanders in Section \ref{sec:Relations}.
Section\ \ref{sec:wireless} is devoted to our new notion of wireless expanders,
where we present nearly tight characterization for the relation between ordinary expanders and wireless expanders.
We start (Section\ \ref{frame}) with describing our basic framework; the positive and negative results are presented in Section\ \ref{sec:positive}
and Section\ \ref{sec:neg}, respectively.
(As mentioned, some positive results are deferred to the appendix. These improve on the parameters provided in Section \ref{sec:positive} by constant factors,
using explicit deterministic arguments.)
Our results for the Spokesman Election problem \cite{chlamtac1991wave}
are given in Section \ref{rel}.
Finally, Section\ \ref{alt} is devoted to our alternative lower bound proof of $\Omega(D \log(n/D))$ on the broadcast time in radio networks.
%In this section we provide a simple proof for obtaining a tight lower bound of $\Omega(D \log(n/D))$ on the broadcast time in radio networks,
%which holds both in expectation and with high probability.

\newpage
\section{Preliminaries} \label{sec:prel}

%%%%%%%%%%%%%%%%
\subsection{Graph Notation}
\label{Notation and Definitions}
%%%%%%%%%%%%%%%%%%%%%%%%%%%
For an undirected graph $G=(V,E)$, vertex $v \in V$ and a subset $S \subseteq V$, denote the set of $v$'s neighbors in $G$ by $\Gamma(v)=\{u ~\mid~ (u,v) \in E\}$, and let $\Gamma(S)=\bigcup_{v \in S}  \Gamma(v)$ be the \emph{neighborhood} of a vertex set $S$ in $G$ (including neighbors that belong to $S$ itself), and $\EN=\Gamma(S)\setminus S$ be the set of neighbors external to $S$. Also define $\Gamma(v,S)=\Gamma(v) \cap S$ as the neighbors of $v$ in the subset $S$.
The \emph{expansion} of $S$ is the ratio $|\Gamma^-(S)| / |S|$.
The \emph{unique-neighborhood} of $S$, denoted by $\UN(S)$, is the set of vertices outside $S$ that have a unique neighbor from $S$. The \emph{unique-neighbor expansion} of $S$ is the ratio $|\UN(S)| / |S|$.
Let $S'$ be an arbitrary subset of $S$. The \emph{$S$-excluding neighborhood} of $S'$, denoted by $\Gamma_S(S')$, is the set of all vertices outside $S$ that have at least one neighbor from $S'$. Similarly, the \emph{$S$-excluding unique-neighborhood} of $S'$, denoted by $\UN_S(S')$, is the set of all vertices outside $S$ that have a unique neighbor from $S'$. In particular,  $\UN(S)=\UN_S(S)$.
The \emph{wireless expansion} of $S$ is the maximum ratio $|\UN_S(S')| / |S|$ over all subsets $S'$ of $S$.
For two sets $S,T \subset V$, let $e(S,T)$ be the set of edges connecting $S$ and $T$.
For vertex $v\in V$, let $\deg(v) = \deg_G(v)$ denote the degree of $v$ in $G$, i.e., the number of $v's$ neighbors, and let $\Delta(G)= \max\{\deg(v)~|~ v \in V\}$ be the maximum degree over all the vertices in $G$. For set $S\subset V$ and vertex $v\in V$, let $\deg(v,S)=\deg(v)\cap S$ be the number of $v's$ neighbors that are in $S$. For two vertices $v,u \in V$, let $d(u,v)$ be the distance between $u$ and $v$ (i.e., the length of the shortest path connecting them), and let $D=D(G)=\max \{d(u,v)~|~ u,v \in V\}$ be the diameter of the graph, i.e. the maximum distance between any two vertices.

We use the combinatorial definition for (vertex) expansion, which requires that
every (not too large) set of vertices of the graph has a relatively large
set of neighbors. Specifically, an $n$-vertex graph $G$ is called an \emph{$(\alpha,\beta)$ vertex expander}
 for positive parameters $\alpha$ and $\beta$, if every subset $S \subseteq V$ s.t. $|S|\leq \alpha n$ has many external neighbors, namely, $|\EN|\geq \beta \cdot |S|$. The \emph{(ordinary) expansion} $\beta(G)$ of $G$ is defined as
the minimum expansion over all vertex sets $S \subseteq V$ of size $|S| \le \alpha n$, namely,
$\beta(G) ~=~ \min\{|\EN|/|S| ~\vert~ S \subseteq V, |S| \le   \alpha n\}.$
A similar definition appears in the literature for bipartite graph, namely, a bipartite graph $G=(L,R, E)$ with sides $L$ and $R$, such that every edge from $E \subset L\times R$ connects one vertex of $L$ and one vertex of $R$ is called an $(\alpha,\beta)$ bipartite vertex expander if every subset $S\subset L$ s.t. $|S|\leq \alpha |L|$ has at least $\beta|S|$ neighbors in  $R$. It is usually assumed that the two sides $L$ and $R$ of the bipartition are of (roughly) the same size.
%%%%%%%%%%%%%%%%

A graph $G=(V,E)$ has \emph{arboricity} $\eta = \eta(G)$ if $$\eta~=~ \max_{U\subseteq V}\left\lceil\frac{|E(U)|}{|U|-1}\right\rceil,$$
where $E(U)=\left\{(u,v)\in E\mid u,v\in U\right\}$. Thus the arboricity is the same (up to a factor of 2) as the maximum
\emph{average degree} over all induced subgraphs of $G$. It is easy to see that for any $(\alpha,\beta)$-expander with maximum degree $\Delta$, the arboricity is
at least $\min\{\Delta / \beta, \Delta \cdot \beta\}$.

%$G[U] = (U,E(U))$

\subsection {Unique Neighbor and Wireless Expanders}
Let us now define formally the notions of unique and wireless expanders. Let $G=(V,E)$ be an $n$-vertex undirected graph. We say that $G$ is an $(\alpha_u, \beta_u)$-{\em unique expander} \cite{AC02} if for every $S \subseteq V$ s.t. $|S|\leq \alpha_u n$, there are at least $\beta_u \cdot |S|$ vertices in $V\backslash S$ that are adjacent to exactly one vertex in $S$, namely, $|\UN(S)|\geq \beta_u\cdot |S|$. The \emph{unique-neighbor expansion} $\beta_u(G)$ of $G$ is defined as
the minimum unique-neighbor expansion over all vertex sets $S \subseteq V$ with $|S| \le \alpha_u n$, namely,
$$\beta_u(G) ~=~ \min\{|\UN(S)|/|S| ~\vert~ S \subseteq V, |S| \le   \alpha_u n\}.$$

We say that $G$ is an $(\alpha_w, \beta_w)$-{\em wireless expander} if for every $S \subseteq V$ s.t. $|S|\leq \alpha_w n$, there exists a subset $S'\subseteq S$  s.t. there are at least $\beta_w \cdot |S|$ vertices in $V\backslash S$ that are adjacent in $G$ to exactly one vertex in $S'$, i.e., $|\UN_S(S')|\geq \beta_w\cdot |S|$. The \emph{wireless expansion} $\beta_w(G)$ of $G$ is defined as
the minimum wireless expansion over all sets $S \subseteq V$ with $|S| \le \alpha_w n$, namely,
$$\beta_w(G) ~=~ \min\{\max\{|\UN_S(S')|/|S| ~\vert~ S' \subseteq S\} ~\vert~ S \subseteq V, |S| \le   \alpha_w n\}.$$
In our arguments, we usually fix $\alpha$ and study the relations between the $\beta$-values for different notions of expanders.
The following connection is easy to verify.

 \begin {observation} \label{relations between the betas}
If $\alpha=\alpha_u= \alpha_w$, then $\beta(G)\geq \beta_w(G)\geq\beta_u(G)$.
\end {observation}

%%%%%%%%%%%%%%%%%%%%%%%%%%%%%%%%
%\vspace{-5pt}

%{Basic bounds on expansion}
%\subsection
\section{Relations between $\beta$ and $\beta_u$}
\label{sec:Relations}
%%%%%%%%%%%%%%%%%%%%%%%%%%%%%%%%

Let $G=(V,E)$ be a $d$-regular undirected graph and let $A=A_G=(a_{uv})_{u,v\in V}$ be its adjacency matrix given by $a_{uv}=1$ if $(u,v)\in E$ and $a_{uv}=0$ otherwise. Since $G$ is $d$-regular, the largest eigenvalue of $A$ is $d$, corresponding to the all-$1$ eigenvector (as $1/d\cdot A$ is a stochastic matrix). Let $\lambda=\lambda_2$ denote the second largest eigenvalue of $G$.
%The proof of the following lemma is deferred to Appendix \ref{app:relation}.
\begin{lemma}
\label{lem:uniq1}
If  a $d$-regular graph $G=(V,E)$ is an $(\alpha_u, \beta_u)$-unique expander,
then it also an $(\alpha, \beta)$-expander with $\alpha= \alpha_u$ and
$\beta\geq (1-1/d) \cdot\beta_u + (d- \lambda) \cdot (1-\alpha_u)/d$.
\end{lemma}
%\APPENDA
%\def\APPENDA{
\begin{proof}
Alon and Spencer \cite{AlonS92} prove that
every partition of the set of vertices $V$ into two disjoint subsets $A$ and $B$ satisfies $|e(A,B)|\geq (d-\lambda)\cdot|A|\cdot|B|/|V|$. In our case (i.e. $A=S,~ B=V\setminus S= \bar{S}$,~ and $|S|\leq\alpha_u\cdot|V|$) we get that
\begin{eqnarray*}
|e(S,\bar{S})| &\geq& (d-\lambda)\cdot\frac{|S|\cdot|\bar{S}|}{|V|}
\\&\geq& (d-\lambda)\cdot\frac{|S|\cdot(|V|-\alpha_u\cdot |V|)}{|V|} \\
&=& (d-\lambda)\cdot|S|\cdot(1-\alpha_u).
\end{eqnarray*}
Moreover, by the expansion properties, there exists a set $U$ of at least $\beta_u\cdot|S|$ vertices in $\EN$ that have a unique neighbor in $S$. From uniqueness, we have $e(S,U)=|U|\geq \beta_u|S|$. Thus, there are at least $(d-\lambda)\cdot|S|\cdot(1-\alpha_u)-|U|$ edges in $e(S,\bar{S})$ that are not connect to the vertices in $U$ (i.e. in $e(S, \bar{S}\setminus U)$). Now, because $G$ is $d$-regular, we get that there exist at least $|U|+((d-\lambda)\cdot|S|\cdot(1-\alpha_u)-|U|)/d$ vertices in $\EN$. Hence, we get \begin{eqnarray*}
|\Gamma^-(S)| &\geq& |U|+\frac{((d-\lambda)\cdot|S|\cdot(1-\alpha_u)-|U|)}{d}\\
&=& \left(1-\frac{1}{d}\right)\cdot |U|+\frac{(d-\lambda)\cdot(1-\alpha_u)}{d}\cdot|S| \\
&\geq& \left(1-\frac{1}{d}\right)\cdot \beta_u |S|+\frac{(d-\lambda)\cdot(1-\alpha_u)}{d}\cdot|S|\\
&=& \left(\left(1-\frac{1}{d}\right)\cdot \beta_u+\frac{(d-\lambda)\cdot(1-\alpha_u)}{d}\right)\cdot|S|
\end{eqnarray*}
thus, $G$ is a $(\alpha,\beta)$-expander with $\alpha\geq\alpha_u$ and $\beta\geq (1-1/d)\cdot\beta_u+((d-\lambda)\cdot(1-\alpha_u)/d$.
\QED
\end{proof}

It is known (and easy to verify) that ordinary expanders whose expansion
is close to the (maximum) degree in the graph are also good unique expanders, or formally:
%The following lemma is provided for completeness.
%
%{\bf DP: Check if this lemma indeed appears in \cite{AC02}}

\begin{lemma}
%{\bf \cite{AC02}}
\label{unique_lowerbound1}
Suppose $G=(V,E)$ is an $(\alpha,\beta)$-expander with maximum degree $\Delta$. Then it is also a unique $(\alpha_u ,\beta_u )$-expander, with $\alpha_u=\alpha$
and $\beta_u\geq 2\beta-\Delta$.
\end{lemma}
{\bf Remark.} Substituting $\beta = (1-\varepsilon) \Delta$
(for $\varepsilon \le 1/2$), we obtain $\beta_u \ge  (1-2\varepsilon) \Delta$.

\commabs
\begin{proof}
Let $S$ be a vertex set with $|S|\leq\alpha\cdot|V|$. Then
$|\EN|\geq\beta\cdot|S|$. Let $\Gamma_0= \EN\setminus\UN(S)$ be the set
of vertices from $\EN$ that have at least two neighbors from $S$.
We have $|\UN(S)|+|\Gamma_0| = |\EN|\geq\beta\cdot|S|$. The number of edges between $S$ and $\EN$ satisfies $e(S,\EN)\leq\Delta\cdot|S|$ on the one hand, but $e(S,\EN)\geq|\UN(S)| + 2|\Gamma_0|$ on the other. Hence $|\Gamma_0| + \beta\cdot|S|\leq|\Gamma_1(S)| + 2|\Gamma_0|\leq \Delta\cdot|S|$, yielding $|\Gamma_0|\leq (\Delta - \beta)\cdot|S|$, and so \begin{eqnarray*}
|\UN(S)| = |\EN| - |\Gamma_0|\geq\beta\cdot|S| - (\Delta - \beta)\cdot|S| = (2\beta - \Delta)\cdot|S|. \inQED
\end{eqnarray*}
\end{proof}
\commabsend

The lower bound $2\beta-\Delta$ on the unique-neighbor expansion $\beta_u$
provided by Lemma \ref{unique_lowerbound1} is meaningful only when $\beta$ is
larger than $\Delta/2$. The following example shows that this lower bound
$2\beta-\Delta$ is  tight.
\begin{lemma} \label{bad unique expander}
For any $\Delta$ and $\beta$ such that $\Delta/2 \le \beta \le \Delta$, there is an $(\alpha, \beta)$ bipartite expander $G^{bad}=(S,\NN,E)$ with
maximum degree $\Delta$ whose unique expansion is $\beta_u\leq 2\beta-\Delta$.
\end{lemma}
\begin{proof}
Construct the graph $G^{bad}$ as follows. Let $S = \{v_1,\ldots,v_s\}$, with $s = |S|$, and suppose that each vertex $v_i \in S$ has exactly $\Delta$ neighbors, all of which are in $\NN$.
(For technical convenience, we define $ v_0 = v_s$, $ v_1 =v_{s+1}$; that is, the vertices $v_1$ and $v_s$ are not different than the other vertices
(they should not be viewed as ``endpoints'', but rather part of an implicit ``cycle'').
Moreover, for each $i = 1,\ldots,s$, the vertices $v_i$ and $v_{i+1}$ have exactly $\Delta-\beta$ common neighbors;
that is, $|\Gamma(v_i) \cap \Gamma(v_{i+1})| = \Delta-\beta$.
More concretely, writing $\Gamma(v_i) = \{v^1_i,\ldots,v^\Delta_i\}$, we have that $$\Gamma(v_i) \cap \Gamma(v_{i+1}) =
\{v^{\beta+1}_i,\ldots,v^\Delta_i\} = \{v^1_{i+1},\ldots,v^{\Delta-\beta}_{i+1}\}.$$
In other words, the ``last'' $\Delta-\beta$  neighbors $v^{\beta+1}_i,\ldots,v^\Delta_i$ of $v_i$ are the ``first'' $\Delta-\beta$  neighbors $v^1_{i+1},\ldots,v^{\Delta-\beta}_{i+1}$ of $v_{i+1}$, respectively.  (See Figure \ref{fig:unique1} for an illustration.)

%%%%%%%%%%%%%%%%
\begin{figure}[htb]
\begin{center}
\includegraphics[scale=0.68]{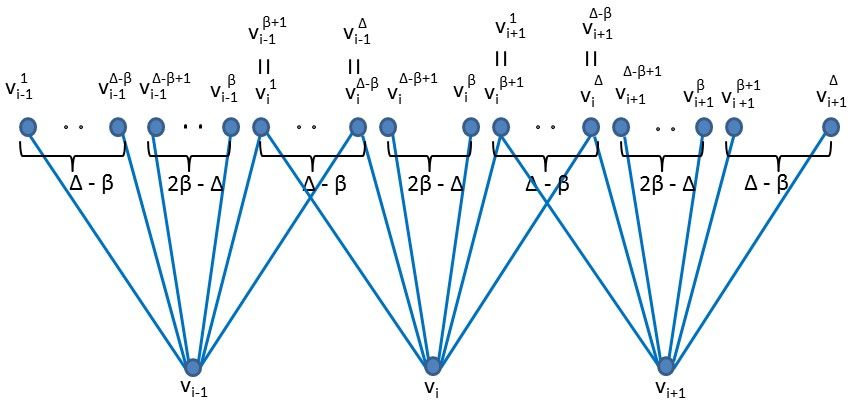}
\end{center}
\caption{\sf An illustration of a worst-case scenario for the unique-neighbor expansion.}
\label{fig:unique1}
\end{figure}
%%%%%%%%%%%%%%%

This means that for each $i = 1,\ldots,s$, the first (resp., last) $\Delta-\beta$ neighbors of $v_i$ are also neighbors of $v_{i-1}$ (resp,  $v_{i+1}$).
%whereas its last $d-\beta$ neighbors are also neighbors of.
The remaining $\Delta - 2(\Delta-\beta) = 2\beta - \Delta$ neighbors of $v_i$, however, are uniquely covered by $v_i$.
%As for vertices $v_1$ and $v_s$, exactly $\beta$ of their neighbors are uniquely covered.
It follows that the number of vertices in the neighborhood of $S$ that are uniquely covered by vertices from $S$ is equal to $s(2\beta-\Delta)$.
Consequently, the unique neighbor expansion $\beta_u$ is $2\beta-\Delta$, as claimed.
Noting that the ordinary expansion is $\beta$ completes the proof of the lemma. \QED
\end{proof}
%}%APPENDB

{\bf Remarks.~}
(1) The meaning of Lemma \ref{bad unique expander} is that a graph with high (ordinary) expansion may have unique neighbor expansion of zero. For example, in the graph $G^{bad}$ described in the proof of Lemma \ref{bad unique expander}, the unique-neighbor expansion is $2\beta -\Delta$, but the wireless expansion is at least $\max\{2\beta-\Delta, \Delta/2\}$.
To see that, let $S'$ be a subset of $S$ and suppose $S'=S_1\cup S_2\cup \ldots \cup S_k$ such that each $S_i$ is a sequence of consecutive vertices, i.e., using the previous notations, for $S_i$ of size $l$, $S_i=\{v_j,..,v_{j+l}\}$ for some index $1\leq j\leq s$. Suppose also that between every two sets $S_i$ and $S_j$ there is at least one vertex that is not in $S'$ (in other words, we can't expand $S_i$ to be a longer secuence in $S'$). Therefore, to compute $\beta_w$, it is enough to compute the expansion parameter for each $S_i=\{v_j,\ldots,v_{j+l}\}$. Consider two options for choosing the set $S''\subset S_i$. The first choice is to take $S''=S_i$. Then we get an expansion of $f(l)=(l\Delta-2(l-1)(\Delta-\beta))/l=((2-l)\Delta+2(l-1)\beta)/l$. The second choice is to take into $S''$ every second vertex in the sequence of $S_i$. Then we get an expansion of $g(l)=l\Delta/(2l)$ if $l$ is even, and $g(l)=(l+1)\Delta/(2l)$ if $l$ is odd. (In the case where $S'=S$ we get in the first choice an expansion of $f(l)=l(2\beta-\Delta)/l=2\beta-\Delta$ and in the second an expansion of $g(l)=(l-1)\Delta/(2l)$). Thus, $\beta_w\geq \min\{\max\{g(l)~,~f(l)\}\mid l>0\}$. As $f(l)$ and $g(l)$ are both decreasing functions, we get that $\beta_w\geq \max\{\lim_{l \to \infty} g(l)~,~\lim_{l \to \infty} f(l)\}=\max\left\{2\beta-\Delta~,~ \frac{\Delta}{2}\right\}.$
This calculation also shows that if $\beta=\Delta/2$, then  the unique-neighbor expansion becomes 0, but the wireless expansion becomes $\Delta/2$.
\vspace{6pt}
\\
\noindent
(2) Although the bipartite graph used in the proof of Lemma \ref{bad unique expander} is an ordinary bipartite expander (according to the definition given in Section \ref{Notation and Definitions}),
note that the sizes of the two sides $S$ and $N$ differ by a factor of $\beta$.
%$\NN$ is greater than that of $S$ by a factor of $\log 2|S|$.
Also, it does not provide an ordinary non-bipartite expander, because the expansion is achieved only on one side, from $S$ towards $N$.
Nevertheless, one can plug this ``bad'' bipartite graph on top of an ordinary $(\alpha,\beta)$-expander with a possibly good unique-neighbor expansion, so that the graph resulting from this tweak
is an ordinary $(\alpha,\beta)$-expander with a unique-neighbor expansion bounded by $2\beta - \Delta$.
Notice, however, that the maximum degree in the resulting graph, denoted by $\Delta'$, may be as large as the sum of the maximum degrees of the ``bad'' bipartite graph
and the $(\alpha,\beta)$-expander that we started from. For example, if $\Delta' = 2\Delta$, then the unique-neighbor expansion of the resulting graph is bounded by $2\beta - \Delta = 2\beta - \Delta'/2$.
%Plugging the generalized core graph on top of $G$ should be done carefully, otherwise
%A-priori, it is unclear that one may significantly reduce the wireless expansion of the graph without reducing the ordinary expansion at a similar rate.
%In particular, it is not enough to guarantee that the core graph has good expansion, because this guarantee does not apply to the vertices of $\NN$.
%This is where the inverse-expansion of the core graph (as defined above) comes into play.
Since we apply a similar tweak in Section \ref{sec:neg} (in the context of wireless expansion rather than unique expansion),
we omit the exact details of this rather simple tweak from the extended abstract.
%\vspace{-5pt}
%%%%%%%%%%%%%%%%%%%%%%
\section{Bounds on Wireless Expansion} \label{sec:wireless}

\subsection{Our Framework} \label{frame}

Consider an arbitrary (ordinary) $(\alpha,\beta)$-expander $G$.
As shown in Section \ref{sec:Relations}, the unique-neighbor expansion $\beta_u$
provided by $G$ may be zero even if the ordinary expansion $\beta$ is high.
In what follows we demonstrate that the wireless expansion $\beta_w(G)$ of $G$
cannot be much lower than its ordinary expansion $\beta(G$).
Moreover, we prove asymptotically tight bounds on the ratio
$\beta(G) / \beta_w(G)$.
This yields a strong separation between the unique-neighbor expansion and
the wireless expansion, which provides a natural motivation for studying
wireless expanders, particularly in applications
where we are given a \emph{fixed} expander network (that cannot be changed).

First let us observe that by Obs. \ref{relations between the betas},
Lemma \ref{unique_lowerbound1} yields the following bound on $\beta_w$.
\begin{lemma}\label{from unique to wireless}
Suppose $G=(V,E)$ is an $(\alpha,\beta)$-expander with maximum degree $\Delta$.
Then it is also a wireless $(\alpha_w ,\beta_w )$-expander
with $\alpha_w=\alpha$ and $\beta_w\geq 2\beta-\Delta$.
\end{lemma}

Throughout what follows, we simplify the discussion by focusing attention to an arbitrary bipartite graph $G_S = (S, \NN, E_S)$ with sides $S$ and $\NN$, such that $|\NN| \geq \beta \cdot |S|$.
We assume that no vertex of $G_S$ is isolated, i.e., all vertex degrees are at least 1.

Note that this bipartition can be thought of as representing all edges in the original graph $G$ that connect an arbitrary vertex set $S$ with its neighborhood $\NN = \EN$. While in $G$ there might be edges internal to $S$ and/or $\NN$, ignoring these edges has no effect whatsoever on the expansion bounds.
\par Our goal is to show the existence of a subset $S'$ of $S$ in the graph $G_S$, whose $S$-excluding unique-neighborhood $\UN_{S}(S')$ is not much smaller than the entire neighborhood $\NN$ of $S$.
%Fix an arbitrary vertex set $S$ in $G$ with $|S| \le \alpha n$.
%From now on we restrict our attention to the \emph{induced} bipartite subgraph $G_S = (S \cup \Gamma, E_S)$ of $G$ with sides $S$ and $\Gamma := \Gamma(S)$,
%where $E(S) = \{(u,v) ~\vert~ u \in S, v \in \Gamma\}$.
%Since the expansion of $S$ is at least $\beta$, hence $|\Gamma(S)| \ge \beta \cdot |S|$.  %, or equivalently $|\Gamma(S)| / |S| \ge \beta$.
Of course, this would imply that the wireless expansion of an arbitrary set $S$ in $G$ (of any size) is close to its ordinary expansion,
yielding the required result.
%In what follows we define $s=|S|$, $\gamma= |N|$.
%%%%%%%%%%%%%%%%
%\vspace{-5pt}
\subsection{Positive Results: Ordinary Expanders are Good Wireless Expanders}
\label{sec:positive}
Let $\delta_S$ (resp., $\delta_N$) be the average degree of the set $S$ (resp., $N$) in the graph $G_S$. That is, $\delta_S=\sum_{u \in S}\deg(u,N)/|S|$ and $\delta_N=\sum_{u' \in N}\deg(u',S)/|N|$. Clearly, $\delta_N, \delta_S \ge 1$.
In this section, we show that $\beta_w$ can be bounded from below as a function of $\min\{\delta_S,\delta_N\}$.

We begin by considering an $(\alpha,\beta)$-expander $G$ for $\beta\geq 1$. We now show:
\begin{lemma}
\label{lem:logdelta}
For every $\beta,\Delta\geq 1$, there exists a subset $S^* \subseteq S$, satisfying that \\$|\UN_{S}(S^*)|=\Omega(|N|/\log 2\delta_N)=\Omega(\beta / \log 2\delta_N)\cdot |S|$. Hence, $\beta_w=\Omega(\beta / \log 2(\Delta/\beta))$.
\end{lemma}
\begin{proof}
Since $\beta\geq 1$, we have $|S|\leq |N|$, $1 \le \delta_N\leq \delta_S$ and $\delta_N \leq \Delta/\beta$.
The proof relies on the probabilistic method.
First, consider the set $N'$ of all vertices from $N$ with degree at most $2\delta_N$. Note that $|N'|\geq |N|/2$ and that all vertices of $N'$ have positive degree.
We now divide the subset $N'$ into $k=\lfloor \log 2\delta_N \rfloor$ subsets depending on their degree in $S$, where the $i^{th}$ subset $N_i$ consists of all vertices $u \in N'$ with $\deg(u,S)\in [2^{i},2^{i+1})$. Let $N_{j}$ be the largest subset among these $k$ subsets. We have that $|N_{j}|\geq |N|/k=\Omega(|N|/\log 2\delta_N) = \Omega(|N| / \log 2(\Delta/\beta))$.
We next show that there exists a subset $S^* \subseteq S$ such that $\UN_{S}(S^*)$ contains a constant fraction of the vertices of $N_j$.

Consider a random subset $S' \subseteq S$ obtained by sampling each vertex $u \in S$ independently with probability $1/2^{j}$.
 For every vertex $u \in N_{j}$, let $X(u) \in \{0,1\}$ be the indicator random variable that takes value 1 if $u$ has exactly one neighbor in $S'$. As $\deg(u,S) \in [2^j,2^{j+1})$, we have that
\begin{eqnarray*}
\Exp(X(u))&=&\Prob(X(u)=1)=\deg(u,S)/2^{j} \cdot (1-1/2^{j})^{\deg(u,S)-1}
\\&\geq&
 (1-1/2^{j})^{2^{j+1}-1}\geq e^{-3}~.
\end{eqnarray*}
Hence, $\sum_{u \in N_{j}}\Exp(X(u))=\Omega(|N_j|)=\Omega(\beta |S|/\log 2(\Delta/\beta))$.
We get that the expected number of vertices in $N$ that are uniquely covered by a random subset $S'$ is $\Omega(\beta |S|/\log 2(\Delta/\beta))$. Hence, there exists a subset $S^* \subseteq S$ with \\$|\UN_{S}(S^*)|=\Omega(\beta / \log 2(\Delta/\beta)) \cdot |S|$. The lemma follows. \QED
\end{proof}
In Appendix \ref{sec:detbounds}, we provide a sequence of deterministic arguments that obtain better bounds for $\beta_w$ (by constant factors) compared to the probabilistic argument shown above.

We now turn to consider the case $\beta<1$. In this case the bound on the wireless expansion depends on $\delta_S$, namely, on the average degree in the larger set $S$. We show:
\begin{lemma}
\label{lem:reductionsmallb}
For every $\Delta\geq 1$ and $\beta \in [1/\Delta,1)$, there exists a subset $S^* \subseteq S$, satisfying that $|\UN_{S}(S^*)|=\Omega(\beta / \log \delta_S)\cdot |S|$. Since $\delta_S \leq \beta\cdot\Delta$, we have $\beta_w=\Omega(\beta / \log 2(\Delta \cdot \beta))$.
\end{lemma}
\begin{proof}
Let $S' \subseteq S$ be the set of all vertices $u\in S$ with $\deg(u,N)\leq 2\delta_S$, and note that $|S'|\geq |S|/2$.
Let $N'=\Gamma^{-}(S')$ be the set of neighbors of $S'$ in $N$. By the expansion of $G$, we have $|N'|\geq \beta \cdot |S'|\geq \beta|S|/2$. We now claim that there exists a subset $S'' \subseteq S'$ satisfying   $\Gamma^-(S'')=N'$ and $|S''|\leq |N'|$. To see this, initially set $S''$ to be empty. Iterate over the vertices of $S'$ and add a vertex $u \in S'$ to $S''$ only if it covers a new vertex of $N'$ (i.e.,  it has a new neighbor in $N'$ that has not been covered before). Then $|S''|\leq |N'|$ and hence in the induced bipartite graph $G'$ with sides $S''$ and $N'$, the expansion measure $\beta$', with $\beta'=|N'|/|S''|$, is at least $1$. The average degree of a vertex $u \in N'$ in the graph $G'$ is bounded by $|E(G')|/|N'|\leq 2\delta_S \cdot |S''|/|N'|\leq 2\delta_S$. Employing the argument of Lemma \ref{lem:logdelta} on the bipartite graph $G'$, we get that
there exists a subset $S^* \subseteq S''$ satisfying $|\UN_{S''}(S^*)|=\Omega(|N'| / \log 4\delta_S)=\Omega(\beta/ \log 2\delta_S)|S|$. Since $\delta_S \leq \Delta \cdot \beta$, it follows that $\beta_w=\Omega(\beta / \log 2(\Delta \cdot \beta))$.  \QED
\end{proof}
Theorem \ref{thm:lb} follows from Lemmas \ref{lem:logdelta} and \ref{lem:reductionsmallb}.

%\\
%\vspace{10pt}
%\noindent
\subsubsection{Relation to the Spokesman Election problem \cite{chlamtac1991wave}} \label{rel}
%\subsection{Positive results}
%[[S: need to rewrite]]
Motivated by broadcasting in multihop radio networks, Chalmtac and Weinstien \cite{chlamtac1991wave} defined the \emph{spokesmen election problem}. In this problem, given a bipartite graph $G=(S,N,E)$, the goal is to compute a subset $S' \subseteq S$ with the maximum number of unique neighbors $\UN(S')$ in $N$.
This problem was shown in \cite{chlamtac1985broadcasting} to be NP-hard. In \cite{chlamtac1991wave}, an approximation scheme is presented that computes a subset $S'\subseteq S$ with $|\UN(S')|\geq |N|/\log |S|$, and this approximation scheme was then used to   devise efficient broadcasting algorithms for multihop radio networks.

The bounds provided in Lemmas \ref{lem:logdelta} and \ref{lem:reductionsmallb} refine and strengthen upon the bound of \cite{chlamtac1991wave}.
%Although $\UN(S')$ can be smaller than $N$,
Our bounds show that $|\UN(S')|$ cannot be smaller than $|N|$ by more than a factor that is logarithmic in $2\min\{\delta_N,\delta_S\}$,
which depends on the \emph{average degree} in $G$, whereas the bound of \ref{lem:reductionsmallb} did not preclude the possibility of $|\UN(S')|$ being smaller than $|N|$
by a factor of $\log |S|$.
%these only showed there could be a deviation that is logarithmic in $|S|$.  % rather than on the .  % rather than on the size of its vertex set or even the maximum degree.
Note that $\min\{\delta_N,\delta_S\}$ is always upper bounded by $|S|$, but can be much smaller than it.
In particular, $\min\{\delta_N,\delta_S\}$ is always low in low arboricity graphs (even if the maximum degree is huge), regardless of  $|S|$.  % is arbitrarily large.

We remark that our randomized approach of choosing the subset $S'\subseteq S$ is extremely simple, and in particular, it yields a much simpler solution to the Spokesman Election problem than that of \cite{chlamtac1991wave}.
Since the solution to this problem was used in \cite{chlamtac1991wave} to devise efficient broadcasting algorithms for multihop radio networks,
our solution can be used to obtain simpler broadcasting algorithms for multihop radio networks than those of  \cite{chlamtac1991wave}.

In the next section (Section \ref{sec:neg}), we show that our positive results for $(\alpha,\beta)$-expanders are essentially the best that one can hope for, by providing a ``bad'' expander example. A bad graph expander example for the related Spokesman Election problem was given in \cite{chlamtac1991wave},
but our graph example is stronger than that of
 \cite{chlamtac1991wave} in several ways, and is based on completely different ideas.
%  (our construction, though, is totally different than that of \cite{chlamtac1991wave}). Specifically,
The graph example of \cite{chlamtac1991wave} is tailored for the somewhat degenerate case where $|N|=\Omega(|S|!)$, whence $N$ is exponentially larger than $S$,
thus the expansion of the bad graph (and the degree) is huge.
%The result of \cite{chlamtac1991wave} does not provide any lower bound
%preclude the existence  Whenever the expansion and degree is lower than that, their
%Their example does not provide any lower bound for expanders in which the degree is should have a reasonably low degree, their example is
 %and then the expansion is obviously low.
%the expansion is to get a .
In addition, in their example, one cannot uniquely cover more than $|N|/\log(|S|) = |N|/\log\log |N|$ vertices of $N$, leaving a big gap between their positive and negative results.
%but this does not match their positive result,
%showing that one can always cover $|N|/\log(|N|)$ vertices of $N$.
%Indeed, since in their example $|N|=\Omega(|S|!)$, they only show that one cannot uniquely cover more than $|N|/\log\log |N|$ vertices of $N$,
%and there is a big gap between
%covering $|N|/\log |N|$ vertices in $N$.
Our bad graph example, in contrast, works for any expansion parameter $\beta$. Moreover, similarly to our positive result, the bounds implied by our negative result
depend on  the average degree of the graph rather than the maximum degree or the size of $S$.
In particular, by taking $\beta$ to be constant and $\Delta$ to be sufficiently large, our graph example shows that one cannot cover more than $|N|/\log |N|$ vertices of $N$, which not only matches our positive result, but also closes the gap left by \cite{chlamtac1991wave}.

%Our negative results also refine and strengthen upon negative results for the spokesman election problem from \cite{chlamtac1991wave}, as shown next.
%\subsection{Negative results} \label{nega}
%%%%%%%%%%%%%%%%

\subsection{Negative Results: Worst-Case Expanders} \label{sec:neg}
In this section we present a ``bad graph '' expander construction.  The description of our construction is given in three stages.
%\begin{enumerate}
First, in Section \ref{core} we construct a bipartite graph $G_S = (S, \NN, E_S)$ with sides $S$ and $\NN$
% where $|\NN| \Theta(|S| \log |S|)$.
that satisfies two somewhat contradictory requirements:
On the one hand, for every subset $S'$ of $S$, $|\Gamma(S')| \ge \log 2|S| \cdot |S'|$.
Hence the ordinary expansion of $G_S$, denoted by $\beta$, is at least $\log 2|S|$.
%There is also an expansion in the symmetric direction, from $\NN$ towards $S$, which we refer to as \emph{inverse-expansion},
%and it is at least as large as $1/\beta$.
%That is, for subset $\NN'$ of $\NN$,  $|\Gamma(\NN')| \ge (1/\log 2s) \cdot |\NN'|$
On the other hand, for every subset $S'$ of $S$, $|\UN_{S}(S')| \le (2 / \log 2|S|) \cdot |\NN|$.
Hence the wireless expansion of $G_S$, denoted $\beta_w$, satisfies $\beta_w \le \beta(2 / \log2|S|)$.
Although this graph is an ordinary bipartite expander (according to the definition given in Section \ref{Notation and Definitions}),
note that the size of $\NN$ is greater than that of $S$ by a factor of $\log 2|S|$.
Also, it does not provide an ordinary non-bipartite expander, because the expansion is achieved only on one side, from $S$ towards $\NN$.
Nevertheless, it provides the core of our worst-case expander, and is henceforth referred to as the \emph{core graph}.
Next, in Section \ref{expand} we describe a generalized core graph
$G^*_S=(S^*,\NN^*,E^*_S)$ with an arbitrary expansion $\beta^*$, while preserving the same upper bound on the wireless expansion.
Finally, in Section \ref{plug} we plug the generalized core graph on top of an ordinary expander $G(V,E)$ with a possibly good wireless expansion, such that $\NN^*\subseteq V$ and $S^*\cap V=\emptyset$, and demonstrate that the resulting graph $\tilde G=(V\cup S^*, E\cup E^*_S)$ is an ordinary expander with a similar expansion but a poor wireless expansion.
%Plugging the generalized core graph on top of $G$ should be done carefully, otherwise
%A-priori, it is unclear that one may significantly reduce the wireless expansion of the graph without reducing the ordinary expansion at a similar rate.
%In particular, it is not enough to guarantee that the core graph has good expansion, because this guarantee does not apply to the vertices of $\NN$.
%This is where the inverse-expansion of the core graph (as defined above) comes into play.
While the generalized core graph is bipartite, the ordinary expander $G$ that we started from does not have to be bipartite.
%Our modification can make sure to preserve the ``bipartiteness'', i.e,
If the original expander $G$ is bipartite,
we can ensure that the expander resulting from our modification will  also be bipartite.
%\end{enumerate}

\subsubsection{The Core Graph} \label{core}

\begin{lemma} \label{worstcase1}
For any integer $s \ge 1$, there is a bipartite graph $G_S = (S, \NN, E_S)$ such that:
\begin{enumerate}
\item $s:= |S|$ and $|\NN| = s \log 2s$.
\item Each vertex in $S$ has degree $2s-1$.
%(Hence both the maximum degree $\Delta_S$ and the average degree $\delta_S$ of a vertex in $S$ is $2s-1$.)  %, and the average degree. (This is   the maximum degree in $G^*_S$.)
\item The maximum degree $\Delta_\NN$ of a vertex in $\NN$ is $s$, and the average degree $\delta_\NN$ of a vertex in $\NN$ is at most $2s/\log 2s$.
\item For every subset $S'$ of $S$, $|\Gamma(S')| \ge \log 2s \cdot |S'|$. (Hence the ordinary expansion, denoted $\beta$, is at least $\log 2s$.)
%\item For every subset $\NN'$ of $\NN$, $|\Gamma(\NN')| \ge (1/\log 2s) \cdot |\NN'|$. (Hence the inverse-expansion is at least $1/\log 2s$.)
\item For every subset $S'$ of $S$, $|\UN_{S}(S')| \le 2s = (2 / \log 2s) \cdot |\NN|$.
(Hence the wireless expansion, denoted $\beta_w$, satisfies $\beta_w \le \beta(2 / \log2s)$.)
\end{enumerate}
\end{lemma}
\begin{proof}
%Let
%$\Gamma = \{v_1,\ldots,v_\gamma\}$, $S = \{z_1,\ldots,z_s\}$, where
%$\gamma = |\NN|, s = |S|$, where $s$ is the minimum integer power of 2 greater than $\gamma/(2\log \gamma)$.
%Since $\gamma/(2\log \gamma) < s \le \gamma/\log \gamma$ and $\log \gamma \ge 2$, we have
%\begin{equation} \label{sizeofgamma1} \frac{\gamma}{4} ~<~ \frac{\gamma}{2\log \gamma} \log \left(\frac{\gamma}{\log \gamma}\right) ~<~     s \log 2s
%~\le~ \frac{\gamma}{\log \gamma} \log \left(\frac{2\gamma}{\log \gamma}\right) ~\le~ \gamma.
%\end{equation}
We assume for simplicity that $s$ is an integer power of 2, which may effect the bounds in the statements of the lemma by at most a small constant.
To describe the edge set $E_S$ of $G_S$, consider a perfect binary tree $T_S$ with $s$ leaves (and $s-1$ internal vertices).
We identify each leaf $z$ of $T_S$ with a unique vertex of $S$.
%For an internal vertex $v$ of $T$ with children $v_L$ and $v_R$, we have $S(v) = S(v_L) \cup S(v_R)$.
Each vertex $v$ of $T_S$ is associated with a set $\NN_v$ of vertices from $\NN$; all these vertex sets are pairwise disjoint, and we have $\NN= \bigcup_{v \in T_S} \NN_v$.
For a vertex $v$ at level $i$ of the tree, $i = 0,1,\ldots,\log s$, the set $\NN_v$ contains $s/2^i$ vertices.
Thus the sizes of these vertex sets decrease geometrically with the level, starting with the set $\NN_{rt}$ at the root $rt$  that consists of $s$ vertices,
and ending with singletons at the leaves.
%\\
%\def\TextBox{
Denote by $\NN_i$ the union of the sets $\NN_v$ over all $i$-level vertices in $T_S$. For all $i = 0,1,\ldots,\log s$, we have $|\NN_i| = s$,
hence $|\NN| = s \log 2s$.
%By Equation (\ref{sizeofgamma1}), $\gamma/4 < |\NN| < \gamma$.
For a leaf $z$ in $T_S$, let $A(z)$ denote the set of its ancestors in $T_S$ (including $z$ itself), and let $\hat \NN_z = \bigcup_{w \in A(z)}\NN_w$.
Define $E(z) = \{(z,v) ~\vert~ v \in \hat \NN_z\}$. Then $E_S = \bigcup_{z \in S} E(z)$.
(See Fig. \ref{fig:basiclb-new} for an illustration.)
%(See Figure \ref{fig:basiclb} for an illustration.)
%TextBox
%%%%%%%%%%%%%%%%%%%%%%%%%%%%%%%%%%%%%
%\vspace{-0.3cm}

%\hspace{-16pt}
%%\parbox{3in}{
%\begin{minipage}{0.3\textwidth}    % decrease size of minipage
%\TextBox
%\end{minipage}
%}
%%\
%%\  \
%%\ ~~ \
%\parbox{3.5in}{

%}

%\def\AltMethod{
%\TextBox
%%%%%%%%%%%%%%%%%%%%%%%%%%%%%%%%%%%%%%%%%%
%\setlength{\columnsep}{10pt}%
%\begin{wrapfigure}{r}{6.5cm}
%\centering
\begin{figure}[htb]
\begin{center}
\includegraphics[width=0.85\textwidth]{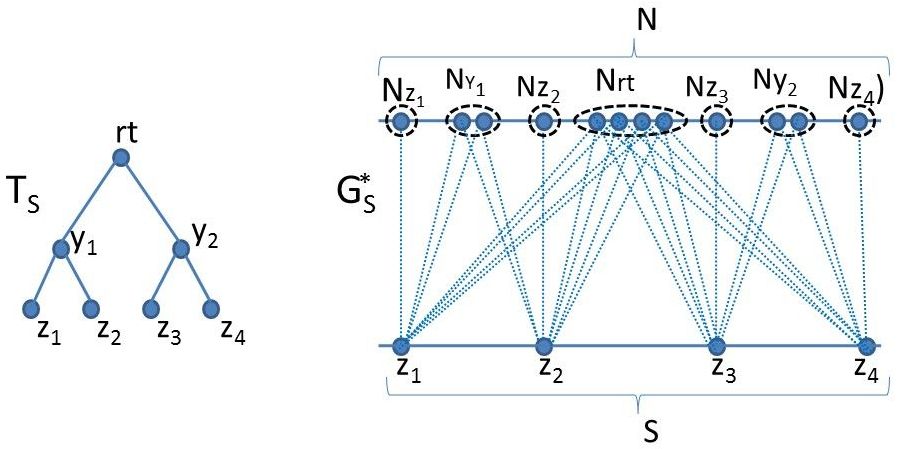}
\end{center}
\caption{
\label{fig:basiclb-new}}
\end{figure}

%	\caption{{\footnotesize An illustration of a graph $G_S$ with worst possible (up to constants) wireless expansion, and the underlying binary tree skeleton $T_S$ used for its construction.}}
%	\label{fig:basiclb}
%\end{wrapfigure}
%%\ignore{
%%%%%%%%%%%%%%%%%
%\begin{figure}[htb]
%\begin{center}
%\includegraphics[scale=0.5]{figs/basiclb.jpg}
%\end{center}
%\caption{\sf An illustration of a graph $G_S$ with worst possible (up to constants) wireless expansion, and the underlying binary tree skeleton $T_S$ used for its construction.}
%\label{fig:basiclb}
%\end{figure}
%%%%%%%%%%%%%%%%
%%}
%
%\setlength{\columnsep}{18pt}%
%\begin{wrapfigure}{r}{7cm}
%\centering
%\includegraphics[width=0.35\textwidth]{figs/basiclb.jpg}
	%\caption{{\footnotesize An illustration of a graph $G_S$ with worst possible (up to constants) wireless expansion, and the underlying binary tree skeleton $T_S$ used for its construction.}}
	%\label{fig:basiclb}
%\end{wrapfigure}
%%%%%%%%%%%%%%%%%%%%%%%%%%%%%%%%%%%%%%%%%%%%%%%%%%%%%%%%%%%%%%%%%%%%%%%%

\begin{observation} \label{ob:basic}
There is an edge between vertex $z \in S$ and vertex $v \in \NN$ iff the unique vertex $w$ in $T_S$ such that $v \in \NN_w$ is an ancestor of $z$ in $T_S$.
\end{observation}
Note that the degree of each vertex $z \in S$, namely $|E(z)|$, is equal to $\sum_{i=0}^{\log s} 2^i = 2s-1$.
On the other hand, the degrees of vertices in $\NN$ are not uniform. For a vertex $v$ in $T_S$, each vertex in $\NN_v$ is incident on  the descendant leaves of $v$.
This means that if $v$ is at level $i$ of $T_S$, then all vertices in $\NN_v$ have degree $2^{\log s-i}=s/2^i$. Hence, the maximum degree $\Delta_\NN$ of a vertex in $\NN$ is $s$ and the average degree $\delta_\NN$ of a vertex in $\NN$ is given by
\begin{eqnarray*}
\delta_\NN &=& \frac{1}{|\NN|} (\sum_{i=0}^{\log s} |\NN_i| (s/2^i))
\\&=& \frac{1}{|\NN|}(\sum_{i=0}^{\log s} \frac{s^2}{2^i})
\le \frac{2s^2}{s\log 2s} ~=~ \frac{2s}{\log 2s}~.
\end{eqnarray*}
Next, we lower bound the   expansion $\beta$ of the graph $G_S$.
Fix an arbitrary set $S' \subseteq S$ of size $k$, for any $1 \le k \le s$, and consider the set of $k$ leaves in $T_S$ identified with $S'$, denoted by $s_1, \ldots,s_k$.
Recall that the level of the root $rt$ is 0, the level of its children is 1, etc., the level of the leaves of $T_S$ is  $\log s$;
in what follows we say that a vertex has \emph{inverse-level} $j$ if its level in $T_S$ is $\log s - j$.
For each vertex $v$ at inverse-level $j$ in $T_S$, the associated vertex set $\NN_v$ has size $2^j$.
Next, we distinguish between inverse-levels at most  $\lfloor \log k \rfloor$ and higher inverse-levels.
For any inverse-level $0 \le j \le \lfloor \log k \rfloor$, the number of ancestors of the $k$ leaves $s_1,\ldots,s_k$ in the tree $T_S$ is at least $k/2^j$,
hence the union of the corresponding vertex sets
%associated with the ancestors of the $k$ leaves $s_1,\ldots,s_k$ inverse-level $j$
 is of size at least $k$. (The lower bound is realized when the $k$ leaves are consecutive to each other in $T_S$.)
For each inverse-level higher than $\lfloor \log k \rfloor$, the number of ancestors of the $k$ leaves $s_1,\ldots,s_k$ may be as small as 1,
but the vertex set associated with such an ancestor is of size at least $k$.
It follows that the union of the corresponding vertex sets at each level is lower bounded by $k$, and so
the union of the vertex sets of all ancestors of the $k$ leaves $s_1, \ldots,s_k$ over all levels is at least $(\log s + 1) \cdot k$.
%is  $\Theta(\log (2\Delta_0/\beta_0))$, i.e., $G^*_S$ is an (ordinary) $(1,\beta)$-expander for $\beta=\Theta(\log (2\Delta_0/\beta_0))$.
%\end{claim}
%\Proof
%Let $T_S$ be the perfect binary tree as described in the construction of $G^*_S$ and let $S=\{s_1,...,s_s\}$.
%Recall that we identify each leaf $z$ of $T_S$ with a unique vertex of $S$. For technical convenience, we define $s_1$ to be the vertex corresponding to the left-most leaf in %$T_S$ and $s_s$ to be the vertex corresponding to the right-most leaf in $T_S$.
By Observation \ref{ob:basic}, all the vertices in this union are neighbors of the vertices in $S'$, thereby yielding $|\Gamma(S')| \ge (\log s + 1) \cdot k = \log 2s \cdot |S'|$.
It follows that $\beta \ge \log 2s$.

\ignore{
We proceed to lower bounding the inverse-expansion of the graph $G_S$.
Fix an arbitrary set $\NN'$ of $\NN$ and some level $i \in \{0,\ldots,\log s\}$, and consider the vertex set $\NN'_i := \NN' \cap \NN_i$, i.e., the intersection of the set $\NN'$ and
all the sets $\NN_v$ over all $i$-vertices in $T_S$. For any such vertex $v$,
denote by $L_v$ the set of vertices in $S$ that are associated with the descendant leaves of $v$ in $T_S$, and note that $|L_v| = |N_v| = s/2^i$.
By Observation \ref{ob:basic}, all vertices in $N_v$ are neighbors of all vertices in $L_v$.
Moreover, for distinct vertices $v$ and $w$ at level $i$, $L_v \cap L_w = N_v \cap N_w = \emptyset$.
It follows that $|\Gamma(\NN'_i)| \ge |\NN'_i|$.
Although the sets $\Gamma(\NN'_i)$ may intersect, $i \in \{0,\ldots,\log s\}$, the $\NN'_i$ sets are vertex-disjoint,
from which we conclude that
$$|\Gamma(\NN')| ~\ge~   (1/\log 2s) \sum_{i \ge 0} |\Gamma(\NN'_i)|
~\ge~  (1/\log 2s) \sum_{i \ge 0} |\NN'_i| ~=~ (1/\log 2s) |\NN'|.$$
}

It remains to upper bound the wireless expansion $\beta_w$ of the graph $G_S$.
Fix an arbitrary set $S' \subseteq S$, and recall that $\UN_{S}(S')$ denotes the set of \emph{all} vertices outside $S$ that have a single neighbor from $S'$.
For a vertex $v$ in $T_S$, let $D(v)$ denote the set of its descendants in $T_S$ (including $v$ itself), and let $\check \NN_v = \bigcup_{w \in D(v)}\NN_w$.
%Recall that the level of $rt$ is 0, the level of its children is 1, etc., the level of the leaves of $T_S$ is  $\log s$;
%As before, we say that a vertex has \emph{inverse-level} $j$ if its level in $T_S$ is $\log s - j$.
We argue that for any vertex $v$ at inverse-level $j$, for $j = 0,1,\ldots,\log s$, it holds that $|\UN_{S}(S') \cap \check \NN_v| \le 2^{j+1}-1$.
%where $j = \log s - i$.
The proof is by induction on  $j$.
\emph{Basis $j = 0$.} In this case $v$ is a leaf, hence $\check \NN_v = \NN_v = \{v\}$, and so
$|\UN_{S}(S') \cap \check \NN_v| \le 1 = 2^{j+1} - 1$.
\emph{Induction step: Assume the correctness of the statement for all smaller values of $j$, and prove it for $j$.}
Consider an arbitrary vertex $v$ at level $j$, and denote its left and right children by $v_L$ and $v_R$, respectively.
Suppose first that $S'$ contains at least one leaf $z_L$ from the subtree of $v_L$ and at least one leaf $z_R$ from the subtree of $v_R$.
By Observation \ref{ob:basic}, every vertex in $\NN_v$ is incident to both $z_L$ and $z_R$, hence no vertex of $\NN_v$ belongs to $\UN_{S}(S')$.
It follows that $\UN_{S}(S') \cap \check \NN_v ~=~ (\UN_{S}(S') \cap \check \NN_{v_L}) \cup (\UN_{S}(S') \cap \check \NN_{v_R}).$
By the induction hypothesis, we conclude that
\begin{eqnarray*}
|\UN_{S}(S') \cap \check \NN_v| &=& |\UN_{S}(S') \cap \check \NN_{v_L}| + |\UN_{S}(S') \cap \check \NN_{v_R}|
\\&\le& 2 \cdot (2^{j} -1) ~\le~ 2^{j+1} - 1~.
\end{eqnarray*}
We henceforth assume that no leaf in the subtree of either $v_L$ or $v_R$, without loss of generality $v_L$, belongs to $S'$.
Hence, by Observation \ref{ob:basic} again, no vertex of $\check \NN_{v_L}$ belongs to $\Gamma(S') \supseteq \UN_{S}(S')$, which gives $\UN_{S}(S') \cap \check \NN_v ~=~ (\UN_{S}(S') \cap \NN_v) \cup (\UN_{S}(S') \cap \check \NN_{v_R}).$
Obviously $|(\UN_{S}(S') \cap \NN_v)| \le |\NN_v| = 2^{j}$.
By the induction hypothesis, we obtain $|\UN_{S}(S') \cap \check \NN_v| ~=~ |\UN_{S}(S') \cap \NN_v| + |\UN_{S}(S') \cap \check \NN_{v_R}|
~\le~ 2^{j} + 2^{j}-1 ~=~ 2^{j+1}-1.$
This completes the proof of the induction.

Since $\check \NN_{rt} = \NN$,
applying the induction statement for the root $rt$ of $T_S$ yields
  $$|\UN_{S}(S')| ~=~ |\UN_{S}(S') \cap \check \NN_{rt}| ~\le~ 2^{\log s+1} -1  ~\le~ 2s ~=~ (2 / \log 2s) \cdot |\NN|.$$
  It follows that $\beta_w \le \beta(2 / \log2s)$, which completes the proof of the lemma.  \QED
 \end{proof}

%{\bf Remark.} The proof of this Lemma reveals additional information: Taking $S'$ to be any singleton of $S$ maximizes $|\UN_{S}(S')|$.
%That is, if $S'$ is an arbitrary singleton of $S$ then we have $|\UN_{S}(S')| = 2s-1$,
%but we have $|\UN_{S}(S')| \le 2s-2$ for any subset $S'$ of $S$ with $|S'| \ge 2$.
%\vspace{-10pt}
\subsubsection{The Core Graph with Arbitrary Expansion} \label{expand}

Notice that the expansion of the graph provided by Lemma \ref{worstcase1} is   logarithmic in the size of its vertex set
and also in the maximum and average degree   (both in $S$ and in $\NN$).
In what follows we show how to construct a \emph{generalized core graph} that has an arbitrary expansion.
\begin{lemma} \label{coregen}
For any integer $\Delta^* \ge 1$ and any $\beta^*$ satisfying $(2e) / \Delta^* \le \beta^* \le \Delta^* / (2e)$ (where $e$ is the base of the natural logarithm),
there exists a bipartite graph $G^*_S = (S^*, \NN^*, E^*_S)$ with sides $S^*$ and $\NN^*$ of maximum degree $\Delta^*$, such that
\begin{enumerate}
\item  $|S^*| \le \Delta^*/2$, $|\NN^*| = \beta^* \cdot |S^*|$.
%$|\NN^*| < (\Delta/2) \cdot \log \Delta$.
%\item   Each vertex in $\check S$ has degree $2s-1$.
%(Hence both the maximum degree $\Delta_S$ and the average degree $\delta_S$ of a vertex in $S$ is $2s-1$.)  %, and the average degree. (This is   the maximum degree in $G^*_S$.)
%\item The maximum degree $\Delta_{\NN}$ of a vertex in $\NN$ is $s \cdot (\log 2s/\beta)$, and the average degree $\delta_{\NN}$ of a vertex in $\NN$ is at most
%$2s/\beta$.
\item For every subset $S'$ of $S^*$, $|\Gamma(S')| \ge \beta^* \cdot |S'|$. (Thus, ordinary expansion is at least $\beta^*$.)
%\item For every subset $\NN'$ of $\NN^*$, $|\Gamma(\NN')| \ge (1/\beta) \cdot |\NN'|$. (Hence the inverse-expansion is at least $1/\beta$.)
\item For every subset $S'$ of $S^*$, \\$|\UN_{S^*}(S')| \le (4 / \log (\min\{\Delta^*/ \beta^*, \Delta^*\cdot \beta^*\})) \cdot |\NN^*|$.
%(Hence the wireless expansion is at most $2 / \log2s$.)
(Hence the wireless expansion, denoted $\beta_w$, satisfies \\$\beta_w \le \beta^*(4 / \log (\min\{\Delta^*/ \beta^*, \Delta^*\cdot \beta^*\}))$.)
\end{enumerate}
\end{lemma}
%\begin{proof}
%\def\APPENDLEMMACORGEN{
To prove Lemma \ref{coregen}, we first present
the following two lemmas which generalize Lemma \ref{worstcase1} to get an arbitrary expansion.
\begin{lemma} \label{worstcase2}
For any integer $s \ge 1$ and any $\beta > \log 2s$, there exists a bipartite graph $\hat G_S = (S, \hat \NN, \hat E_S)$ such that
\begin{enumerate}
\item $s:= |S|$ and $|\hat \NN| = s \cdot \beta$.
\item Each vertex in $S$ has degree $(2s-1) \cdot (\beta / \log 2s)$.
%(Hence both the maximum degree $\Delta_S$ and the average degree $\delta_S$ of a vertex in $S$ is $2s-1$.)  %, and the average degree. (This is   the maximum degree in $G^*_S$.)
\item The maximum degree $\Delta_{\hat \NN}$ of a vertex in $\hat \NN$ is $s$, and the average degree $\delta_{\hat \NN}$ of a vertex in $\hat \NN$ is at most $2s/\log 2s$.
\item For every subset $S'$ of $S$, $|\Gamma(S')| \ge \beta \cdot |S'|$. (Hence the ordinary expansion is at least $\beta$.)
%\item For every subset $\NN'$ of $\hat \NN$, $|\Gamma(\NN')| \ge (1/\beta) \cdot |\NN'|$. (Hence the inverse-expansion is at least $1/\beta$.)
\item For every subset $S'$ of $S$, $|\UN_{S}(S')| \le  2s \cdot (\beta / \log 2 s) = (2 / \log 2s) \cdot |\hat \NN|$.
%(Hence the wireless expansion is at most $2 / \log2s$.)
(Hence the wireless expansion, denoted $\beta_w$, satisfies $\beta_w \le \beta(2 / \log2s)$.)
\end{enumerate}
\end{lemma}
\begin{proof}
We assume for simplicity that $k = \beta / \log 2s$ is an integer,
and modify the construction used to prove Lemma \ref{worstcase1} by creating $k$ copies $v_1, \ldots,v_k$ for each vertex $v$ in $\NN$.
Thus each vertex set $\NN_v$ is ``expanded'' by a factor of $k$; denote the expanded vertex set by $\hat \NN_v$.
The vertex set $\hat \NN$ of $\hat G_S$ is the union of all copies of all vertices in $\NN$, or in other words,
it is the union of all the expanded vertex sets, i.e., $\hat \NN = \bigcup_{v \in T_S} \hat \NN_v$.
The edge set $\hat E_{S}$ of $\hat G_S$ is obtained by translating each edge $(v,u)$ in the original graph $G_S$, where $v \in \NN$,
into the $k$ edges $(v_1,u),\ldots,(v_k,u)$ in $\hat G_S$.
%simply ``expanding'' each vertex set $\NN_v$ by a factor of $\beta / \log 2s$.
%Denoting the \emph{expanded vertex set} by $\hat \NN_v$, the vertex set $\hat \NN$ is obtained as the union of all the expanded vertex sets,
%i.e.,
Other than this modification, the construction remains intact.
Note that $S$ remains unchanged, and the degree of vertices in $\hat \NN$ is the same as the degree of vertices in $\NN$ in the original graph $G_S$ (both the maximum and average degree).
On the other hand, we now have $|\hat \NN| = (s \log 2s) \cdot (\beta / \log 2s) = s \cdot \beta$.
Moreover, the expansion increases from at least $\log 2s$ to at least $\beta$,
%the inverse-expansion decreases from at least $1/\log 2s$ to at least $1/\beta$,
and the degree of vertices in $S$ increases from $2s-1$ to $(2s-1) \cdot (\beta / \log 2s)$.
Finally, note that for every subset $S'$ of $S$, $|\UN_{S}(S')|$ increases by a factor of $\beta / \log 2s$,
hence $|\UN_{S}(S')|$ is at most $2s \cdot (\beta / \log 2 s) = (2 / \log 2s) \cdot |\hat \NN|$,
thus the wireless expansion $\beta_w$ satisfies $\beta_w \le \beta(2 / \log2s)$. \QED
\end{proof}
\begin{lemma} \label{worstcase3}
For any integer $s \ge 1$ and any $\beta \le \log 2s$, there exists a bipartite graph $\check G_S = (\check S, \NN, \check E_S)$ with sides $\check S$ and $\NN$,
%with sides $S_0$ and $\NN_0$,
such that
\begin{enumerate}
\item $|\check S| =  s \cdot (\log 2s / \beta)$ and $|\NN| = s \log 2s$.
\item   Each vertex in $\check S$ has degree $2s-1$.
%(Hence both the maximum degree $\Delta_S$ and the average degree $\delta_S$ of a vertex in $S$ is $2s-1$.)  %, and the average degree. (This is   the maximum degree in $G^*_S$.)
\item The maximum degree $\Delta_{\NN}$ of a vertex in $\NN$ is $s \cdot (\log 2s/\beta)$, and the average degree $\delta_{\NN}$ of a vertex in $\NN$ is at most
$2s/\beta$.
\item For every subset $S'$ of $\check S$, $|\Gamma(S')| \ge \beta \cdot |S'|$. (Hence the ordinary expansion is at least $\beta$.)
%\item For every subset $\NN'$ of $\NN$, $|\Gamma(\NN')| \ge (1/\beta) \cdot |\NN'|$. (Hence the inverse-expansion is at least $1/\beta$.)
\item For every subset $S'$ of $\check S$, $|\UN_{\check S}(S')| \le 2s   = (2 / \log 2s) \cdot |\NN|$.
%(Hence the wireless expansion is at most $2 / \log2s$.)
(Hence the wireless expansion, denoted $\beta_w$, satisfies $\beta_w \le \beta(2 / \log2s)$.)
\end{enumerate}
\end{lemma}
\begin{proof}
We assume for simplicity that $k = \log 2s / \beta$ is an integer,
and modify the construction used to prove Lemma \ref{worstcase1} by creating $k$ copies $v_1, \ldots,v_k$ for each vertex $v$ in $S$.
The vertex set $\check S$ of $\check G_S$ is the union of all copies of all vertices in $S$,
and the edge set $\check E_{S}$ is obtained by translating each edge $(v,u)$ in the original graph $G_S$, where $v \in S$,
into the $k$ edges $(v_1,u),\ldots,(v_k,u)$ in $\check G_S$.
%^Denoting the \emph{expanded vertex set} by $S''_v$, the vertex set $\NN'$ is obtained as the union of all the expanded vertex sets, i.e., $\NN = \bigcup_{v \in T_S} \NN'_v$.
Other than this modification, the construction remains intact.
Note that $\NN$ remains unchanged, and the degree of vertices in $\check S$  is the same as the degree of vertices in $S$ in the original graph $G_S$ (both the maximum and average degree).
%Note that $S$ remains unchanged, and the degree of vertices in $\Gamma$ remains unchanged too.
On the other hand, we now have $|\check S| = s \cdot (\log 2s / \beta)$.
Moreover, the expansion decreases from at least $\log 2s$ to at least $\beta$,
%the inverse-expansion increases from at least $1/\log 2s$ to at least $1/\beta$,
and the degree of vertices in $\NN$ increases by a factor of $\log 2s/\beta$.
%Finally, we argue that the upper bound on the wireless expansion remains unchanged. To see this, note that
%for every subset $S'$ of $\check S$, $|\UN_{\check S}(S')|$ remains at most $2s = (2 / \log 2s) \cdot |\NN|$.
Finally, note that for every subset $S'$ of $\check S$, $|\UN_{\check S}(S')|$ remains at most $2s = (2 / \log 2s) \cdot |\NN|$,
%hence $|\UN_{S}(S')|$ is at most $2s \cdot (\beta / \log 2 s) = (2 / \log 2s) \cdot |\hat \NN|$,
thus the wireless expansion $\beta_w$ remains unchanged, satisfying $\beta_w \le \beta(2 / \log2s)$.
\QED
\end{proof}

We are now ready to complete to proof of Lemma \ref{coregen}.
\begin{proof}[Lemma \ref{coregen}]
Since $\beta^* \le \Delta^* / (2e)$, we may write $\Delta^* = 2s \cdot (\beta^* / \log 2s)$, for $s \ge e$.
Suppose first that $\beta^* > \log 2s$.
In this case we take $G^*_S$ to be the graph provided by Lemma \ref{worstcase2} for $\lceil s \rceil$ and $\beta^*=\beta$;
we assume for simplicity that $s$ is an integer, but this assumption has a negligible effect.
The maximum degree  in the graph is $(2s-1) \cdot (\beta^* / \log 2s)$,
which is bounded by $\Delta^* := 2s \cdot (\beta^* / \log 2s)$.
This in particular yields $\Delta^* \ge 2s$, and so $|S^*| = s \le \Delta^*/2$.
We also have $|\NN^*| = \beta^* \cdot |S^*|$.  % < (\Delta/2) \cdot \log 2S \le \Delta/2 \cdot \log \Delta$.
The second assertion follows immediately from Lemma \ref{worstcase2}(4).  % and \ref{worstcase2}(5), respectively.
It remains to prove the third assertion.
Lemma \ref{worstcase2}(5) implies that for every subset $S'$ of $S^*$,
$|\UN_{S^*}(S')| \le  2s \cdot (\beta^* / \log 2 s) = (2 / \log 2s) \cdot |\NN^*|$.
Observe that $$\min\{\Delta^*/ \beta^*, \Delta^*\cdot \beta^*\} ~=~ \Delta^*/ \beta^* ~=~ 2s / \log 2s ~\le~ 2s.$$
Hence $2 / \log 2s \le 2 / \log (\min\{\Delta^*/ \beta^*, \Delta^*\cdot \beta^*\})$, which implies that
\begin{eqnarray*}
|\UN_{S^*}(S')| &\le&  (2 / \log 2s) \cdot |\NN^*| \\&\le& (2 / \log (\min\{\Delta^*/ \beta^*, \Delta^*\cdot \beta^*\})) \cdot |\NN^*|~.
\end{eqnarray*}
We henceforth assume that  $\beta^* \le \log 2s$.
Since $\beta^* \ge (2e) / \Delta^*$, we may write $\Delta^* = 2s' \cdot (\log 2s'/\beta^*)$, for $s' \ge e/2$.
Next, we argue that $\beta^* \le \log 2s'$.
Since $\beta^* \le \log 2s$ and as $\Delta^*$ is equal to both $2s \cdot (\beta^* / \log 2s)$ and $2s' \cdot (\log 2s'/\beta)$,
it follows that
$$((2s')/(2s)) \log(2s') \log (2s) ~=~ (\beta^*)^2 ~\le~ \log^2 (2s).$$
%\end{equation}
Thus $(2s') \cdot \log (2s') \le (2s) \cdot \log (2s)$,
%\begin{equation} \label{basic1}
%
%\end{equation}
and so $s' \le s$.
Next, we prove that $(2s') / \log (2s') \le (2s) / \log (2s)$
%We claim that $(2s') / \log (2s') \le (2s) / \log (2s)$.
by taking logarithms for both hand sides and noting that the function $f(x) = x - \log x$ is monotone increasing for $x > \log e$
and that $s \ge s' \ge e/2$.
Rearranging, we get \\$(\beta^*)^2 = ((2s')/(2s)) \log(2s') \log (2s) \le \log^2 (2s')$, thus $\beta^* \le \log 2s'$.

In this case we take $G^*_S$ to be the graph provided by Lemma \ref{worstcase3} for $\lceil s' \rceil$ and $\beta^*=\beta$;
we again assume for simplicity that $s$ is an integer, but this assumption has a negligible effect.
The maximum degree in the graph is $\max\{2s'-1, s' \cdot (\log 2s'/\beta)\}$, which is bounded by $\Delta^* := 2s' \cdot (\log 2s'/\beta^*)$.
Note that
%\ge \Delta' \ge s' \cdot (\log 2s'/\beta) \ge 2s'-1$,
 $|S^*| = s' \cdot (\log 2s' / \beta^*) = \Delta^*/2$ and $|\NN^*| = s' \log 2s' = \beta^* \cdot |S^*|$.
%\le (\Delta/2) \cdot \log \Delta$.
%The second assertion follows immediately from Lemma \ref{worstcase2}(4).
The second assertion follows immediately from Lemma \ref{worstcase3}(4).   % and \ref{worstcase2}(5), respectively.
It remains to prove the third assertion.
Lemma \ref{worstcase3}(5) implies that for every subset $S'$ of $S^*$,
$|\UN_{S^*}(S')| \le 2s'   = (2 / \log 2s') \cdot |\NN^*|$.
Observe that $$\min\{\Delta^*/ \beta^*, \Delta^*\cdot \beta^*\} ~\le~ \Delta^* \cdot \beta^* ~=~  2s' \cdot \log 2s'.$$
Hence
\begin{eqnarray*}
2 / \log 2s' &=& 4 / \log ((2s')^2) ~\le~ 4/ \log (2s' \cdot \log 2s') \\&\le&
4 / \log(\min\{\Delta^*/ \beta^*, \Delta^*\cdot \beta^*\}),
\end{eqnarray*}
which implies that
\begin{eqnarray*}
|\UN_{S^*}(S')| &\le&  (2 / \log 2s') \cdot |\NN^*| \\&\le& (4 / \log (\min\{\Delta^*/ \beta^*, \Delta^*\cdot \beta^*\})) \cdot |\NN^*|. \inQED
\end{eqnarray*}
\end{proof}

%\APPENDLEMMACORGEN

%\vspace{-10pt}

\subsubsection{Worst-Case Expanders} \label{plug}

%In Section \ref{basicop} we show how to reduce the wireless expansion of an ordinary expander, while essentially preserving its ordinary expansion.
%In Section \ref{genop} we provide a symmetric modification, where the resulting expander contains a vertex set that can be uniquely covered
%only partially by any other vertex set.
%\subsubsection{Reducing the wireless expansion} \label{basicop}
Let $G$ be an arbitrary $(\alpha,\beta)$-expander on $n$ vertices with maximum degree $\Delta$, and let $0 < \epsilon < 1/2$ be a ``blow-up'' parameter.
That is, $\epsilon$ will determine the extent by which the parameters of interest
blow up due to the modification that we perform on the original graph $G$ to obtain poor wireless expansion.
There is a tradeoff between the wireless expansion and the other parameters:
The stronger our upper bound on the wireless expansion is, the larger the blow-up in the other parameters becomes.

For technical reasons, we require that $\Delta \cdot \beta \ge 1/(1-\epsilon^2)$.
We start by constructing the generalized core graph $G^*_S = (S^*, \NN^*, E^*_S)$ provided by Lemma \ref{coregen} for $\Delta^* = \epsilon \cdot \Delta$
and expansion $\beta^* = \beta/\epsilon$,
%By Lemma \ref{coregen},
thus yielding $|S^*| \le \Delta^*/2 = \epsilon(\Delta/2)$ and $|\NN^*| =  \beta^* \cdot |S^*| = (\beta /\epsilon) \cdot |S^*|$.
Our worst-case expander $\tilde G$ is obtained by plugging $G^*_S$ on top of $G$.
The vertices of $S^*$ are not part of the original vertex set of $G$, but are rather new vertices added to it.
The vertices of $\NN^*$ are chosen arbitrarily from $V(G)$.
%(To optimize, we may choose the ones of minimum degree.)
%This optimization does not help in   general, and will not change the asymptotic bounds in any case.)

\noindent
{\bf Remark.} If $G$ is a bipartite expander, expanding from the left side $L$ to the right side $R$,
and if we want $\tilde G$ to remain bipartite and to expand from $\tilde L$ to $\tilde R$, then $\tilde L$ will be defined as the union of $L$ and $S^*$,
and $\tilde R$ will be defined as the union of $R$ and a dummy vertex set of the same size as $S^*$, to guarantee that $|\tilde L| = |\tilde R|$.

%We start by adding an arbitrary set $\NN$ of $\epsilon(\beta \cdot \Delta/2)$ vertices to $G$.  %, with $s$ to be determined later (but it is close asymptotically to $\Delta$).
%Next, we choose an arbitrary set $S$ of $|\NN|/\beta = \epsilon(\Delta/2)$ vertices in $G$.

\ignore{
Finally, we add the edges of the generalized core graph $G^*_S$ to $G$,
where the vertex sets $S^*$ and $\NN^*$ of $G^*_S$ are determined by choosing a set of $|S^*|$ arbitrary vertices of $S$
and a set of $|\NN^*|$ arbitrary vertices of $\NN$, respectively.
%as they will serve as the vertex sets $S^*$ and $\NN^*$.
%and translate each edge
%Consequently, we may 1-to-1 mappings between the vertices of $S^*$ to the vertices of $S$ and between the vertices of $\
To justify this last step, note that the size of $S$ is at least as large as that of $S^*$.
Moreover, since the expansion of $G$ is $\beta$, we have $|\NN| \ge \beta \cdot |S|$, hence the size of $\NN$ is at least as large as that of $\NN^*$.
Consequently, determining the vertex sets $S^*$ and $\NN^*$ as subsets of $S$ and
$\NN$, respectively, is well-defined.
}
In what follows we analyze the properties of $\tilde G$.
Denoting the number of vertices in $\tilde G$ by $\tilde n$, we have
$n ~\le~ \tilde n ~\le~ n + 2|S^*| ~\le~ n + 2\epsilon(\Delta/2) ~\le~ (1+\epsilon) \cdot n.$
Write $\tilde \Delta = (1+\epsilon) \cdot \Delta$, and note that the maximum degree in $\tilde G$ is bounded by
$\Delta  + \Delta^* ~\le~ \Delta + \epsilon \cdot \Delta ~=~ \tilde \Delta.$

\begin{claim}\label{cl:ordinary}
$\tilde G$ is an ordinary $(\tilde \alpha, \tilde \beta)$-expander, where $\tilde \beta = (1-\epsilon)\cdot \beta, \tilde \alpha = (1-\epsilon) \cdot \alpha$.
\end{claim}
\begin{proof}
Since $\tilde n < (1+\epsilon)\cdot n$
and as $\tilde \alpha = (1-\epsilon) \cdot \alpha$, it follows that
$\tilde \alpha \cdot \tilde n ~\le~ (1-\epsilon)\alpha \cdot (1+\epsilon) \cdot n = (1-\epsilon^2) \alpha \cdot n <  \alpha \cdot n.$
Consider an arbitrary set $X$ of at most $\tilde \alpha \cdot \tilde n  \le \alpha \cdot n$ vertices from $\tilde G$.
By Lemma \ref{coregen}(2), the expansion in $G^*_S$ is at least $\beta^* = \beta/\epsilon$,
hence $|\Gamma^-(X \cap S^*)| \ge (\beta/\epsilon) \cdot |X \cap S^*|$.
%, and let $S' = $.
If
%at least an $\epsilon$-fraction of the vertices of $X$ belong to $S^*$, i.e.,
$|X \cap S^*| \ge \epsilon \cdot |X|$,
then
%Since the expansion of the core graph $G^*_S$ is
 we have
%yields $|\Gamma(S')| \ge (\beta / \epsilon) \cdot |S'|$, for every subset $S'$ of $S^*$.
%we have
$|\Gamma^-(X)| ~\ge~ |\Gamma^-(X \cap S^*)| ~\ge~
(\beta/\epsilon) \cdot |X \cap S^*| ~\ge~ (\beta/\epsilon) \cdot (\epsilon \cdot |X|) ~=~ \beta \cdot |X| ~>~ \tilde \beta \cdot |X|.$
Otherwise, $|X \setminus S^*| \ge (1-\epsilon) \cdot |X|$, and as the expansion in $G$ is at least $\beta$, we have
$|\Gamma^-(X)| ~\ge~ |\Gamma^-(X \setminus S^*)| ~\ge~ \beta \cdot |X \setminus S^*|
~\ge~ \beta \cdot (1-\epsilon) \cdot |X| ~=~ \tilde \beta \cdot |X|. $
\QED
\end{proof}
%}%\APPENDORD

Recall that  $\Delta \cdot \beta \ge 1/(1-\epsilon^2)$, and note that
$\tilde \Delta \cdot \tilde \beta = (1+\epsilon)\Delta \cdot (1-\epsilon)\beta \ge 1$.
We also have that $\tilde \Delta / \tilde \beta > \Delta / \beta \ge 1$. Hence  the term $\log (\min\{\tilde \Delta/ \tilde \beta, \tilde \Delta \cdot \tilde \beta\})$ is non-negative,
and the upper bound $O(\tilde \beta / (\epsilon^3 \cdot \log (\min\{\tilde \Delta/ \tilde \beta, \tilde \Delta \cdot \tilde \beta\})))$ in the following claim is well-defined.
%In Appendix \ref{app:helper}, we show:

\begin{claim}\label{cl:helperwire}
The wireless expansion $\tilde \beta_w$ of $\tilde G$ satisfies
$\tilde \beta_w = O(\tilde \beta / (\epsilon^3 \cdot \log (\min\{\tilde \Delta/ \tilde \beta, \tilde \Delta \cdot \tilde \beta\}))).$
\end{claim}
\begin{proof}
%\Proof
%By definition, the wireless expansion of $\tilde G$ is upper bounded by the wireless expansion of
%any subgraph.
Note that $\tilde \beta_w$ is trivially upper bounded by $\beta$,
thus the claim holds vacuously whenever
$\epsilon^3 \cdot \log (\min\{\tilde \Delta/ \tilde \beta, \tilde \Delta \cdot \tilde \beta\}) < 2$.
We may henceforth assume that
%$1 / (\epsilon^3 \cdot \log (\min\{\tilde \Delta/ \tilde \beta, \tilde \Delta \cdot \tilde \beta\})) = O(1)$, %or equivalently,
$\epsilon^3 \cdot \log (\min\{\tilde \Delta/ \tilde \beta, \tilde \Delta \cdot \tilde \beta\}) ~\ge~ 2,$
which implies that both $\tilde \Delta/ \tilde \beta$ and $\tilde \Delta \cdot \tilde \beta$ are at least $2^{2/\epsilon^3}$.
Since $\epsilon < 1/2$,
it follows that
$$\Delta^* \cdot \beta^* ~=~ \Delta \cdot \beta ~\ge~ (\tilde \Delta/(1+\epsilon))  \cdot (\tilde \beta/(1-\epsilon))
~\ge~ \tilde \Delta \cdot \tilde \beta ~\ge~ 2^{2/\epsilon^3} ~\ge~ 2e$$
and
\begin{eqnarray*}
\Delta^* / \beta^* &=& \epsilon^2 (\Delta / \beta) ~\ge~ \epsilon^2 (\tilde \Delta/(1+\epsilon)) / (\tilde \beta/(1-\epsilon))
\\ &=& \epsilon^2((1-\epsilon)/(1+\epsilon)) \cdot (\tilde \Delta / \tilde \beta) \\&\ge& \epsilon^2((1-\epsilon)/(1+\epsilon)) \cdot 2^{2/\epsilon^3} ~\ge~ 2e.
\end{eqnarray*}
In particular, we have $(2e) / \Delta^* \le \beta^* \le \Delta^* / (2e)$, as required in Lemma \ref{coregen}.
Since all edges adjacent to the vertices of $S^*$ belong to the core graph $G^*_S$ with parameters $\Delta^*$ and $\beta^*$,
Lemma \ref{coregen}(3) implies that for every subset $S'$ of $S^*$,
\begin{eqnarray*}
|\UN_{S^*}(S')| &\le& (4 / \log (\min\{\Delta^*/ \beta^*, \Delta^* \cdot \beta^*\})) \cdot |\NN^*|
\\ &\le& (4(1+\epsilon)/(\epsilon^2(1-\epsilon) \cdot \log (\min\{\tilde \Delta/ \tilde \beta, \tilde \Delta \cdot \tilde \beta\}))) \cdot |\NN^*|.
\\ &\le& (12/(\epsilon^2 \cdot \log (\min\{\tilde \Delta/ \tilde \beta, \tilde \Delta \cdot \tilde \beta\}))) \cdot |\NN^*|.
\\ &=& (12/(\epsilon^3 \cdot \log (\min\{\tilde \Delta/ \tilde \beta, \tilde \Delta \cdot \tilde \beta\}))) \cdot \beta \cdot |S^*|.
\\ &\le& (24/(\epsilon^3 \cdot \log (\min\{\tilde \Delta/ \tilde \beta, \tilde \Delta \cdot \tilde \beta\}))) \cdot \tilde \beta \cdot |S^*|.
\end{eqnarray*}
(It is easily verified that the third and last inequalities hold for $\epsilon < 1/2$.)
The bottom-line constant 24 can be improved; we did not try to optimize it.
\QED
\end{proof}
%(Hence the wireless expansion is at most $2 / \log2s$.)
%\APPENDCLHELPW

We derive the following corollary, which implies the existence of expanders with worst possible wireless expansion.
The bound on the wireless expansion is tight in the entire range of parameters, disregarding constants and dependencies on $\epsilon$.

\begin{corollary} \label{raman}
For any $n,\Delta,\beta$ and $0 < \epsilon <1/2$ such that $\Delta \cdot \beta \ge 1/(1-\epsilon^2)$,
if there exists an ordinary $(\alpha,\beta)$-expander $G$ on $n$ vertices with maximum degree $\Delta$,
%then for any $0 < \epsilon < 1$,
then there exists an $(\tilde \alpha,\tilde \beta)$-expander $\tilde G$ on $\tilde n$ vertices with maximum degree $\tilde \Delta$
and wireless expansion $\tilde \beta_w$, where:
%\begin{enumerate}
(1) $\Delta \le \tilde \Delta \le (1+\epsilon) \cdot \Delta$;
(2) $n \le \tilde n \le (1+\epsilon) \cdot n$;
(3) $\tilde \beta = (1-\epsilon)\cdot \beta$;
(4) $\tilde \alpha = (1-\epsilon)\cdot \alpha$; and
(5) $\tilde \beta_w = O(\tilde \beta / (\epsilon^3 \cdot \log (\min\{\tilde \Delta/ \tilde \beta, \tilde \Delta \cdot \tilde \beta\})))$.
\end{corollary}
One may use Corollary \ref{raman} in conjunction with known constructions of explicit expanders (such as Ramanujan graphs), which achieve near-optimal expansion for any degree parameter.
Taking $\epsilon$ to be a sufficiently small constant thus completes the proof of Theorem \ref{thm:ub}.

\section{A tight lower bound on the broadcast time in radio networks} \label{alt}
In this section we provide a simple proof for obtaining a tight lower bound of $\Omega(D \log(n/D))$ on the broadcast time in radio networks,
which holds both in expectation and with high probability.

Consider our core bipartite graph $G_S = (S, \NN, E_S)$ from Lemma \ref{worstcase1}, with sides $S$ and $\NN$, where $s = |S|$ and $|\NN| = s \log 2s$.
Suppose that we connect an additional vertex $rt$ to all vertices of $S$ and initiate a (radio) broadcast at $rt$ in the resulting graph.
By Lemma \ref{worstcase1}(5),
%it's impossible
one cannot uniquely cover more than $2s$ vertices (i.e., a $(2 / (\log 2s))$-fraction) of $\NN$ using any subset $S' \subseteq S$.
It follows that at any round after the first, the broadcast may reach at most $2s$ new vertices of $\NN$, which yields the following corollary.
\begin{corollary} \label{detcor}
The number of rounds needed for the broadcast to reach a $(2i / (\log 2s))$-fraction of $\NN$ is at least $1+i$, for any $0 \le i \le ((\log 2s) / 2)$.
%Thus the number of rounds needed for the broadcast to reach all vertices in $\NN$ is at least $1 + ((\log 2s) / 2)$.
\end{corollary}

%Corollary \ref{detcor} implies that the number of rounds required to reach all vertices in $\NN$ is at least $1 + ((\log 2s) / 2)$.
%This bound is deterministic. Next, we
Next, we construct a graph $G$ of diameter $\Theta(D)$, for an arbitrary parameter $D = \Omega(\log n)$,
in which the number of rounds needed to complete a broadcast is $\Omega(D \log(n/D))$.
%(We may focus on diameter at least Omega(log n), as there is a lower bound of Omega(log^2 n) by David and others that applies even to radius 2 networks.)

The core graph $G_S$ has $|S| + |\NN| = s(1 + \log 2s) = s (\log 4s)$ vertices.
We take $D/2$ copies of this graph, denoted by $G^1_S, G^2_S, \ldots, G^{D/2}_S$, each containing roughly $n / D$ vertices.
Thus we take $s$ so that $n / D \approx s (\log 4s)$, and so $\log s = \Theta(\log (n/D))$.
Denote the sides of $G^i_S$ by $S^i$ and $\NN^i$.
We connect the root $rt = rt^0$ to all vertices of $S^1$,
and for each   $1 \le i \le D/2$, we randomly sample a vertex from $N^i$, denoted by $rt^i$, and connect it (unless $i = D/2$) to all vertices of $S^{i+1}$.
%Also, let $rt^{D/2}$ be
%And in general, we pick at random a vertex of N^i, denoted by rt^i, and connect it to all vertices of S^{i+1}.
This completes the construction of the graph $G$. It is easy to verify that the diameter of $G$ is $\Theta(D)$, and to be more accurate, the diameter is $D + 2$.
In what follows we assume that none of the processors associated with the vertices of the graph initially have any topological information on the graph (except for its size
and diameter). This rather standard assumption was also required in the proof of Kushilevitz and Mansour \cite{kushilevitz1998omega}.

Consider a broadcast initiated at $rt$. We make the following immediate observation.
\begin{observation} \label{imm}
The message must reach $rt^{i-1}$ before reaching $rt^i$, for $1 \le i \le D/2$.
\end{observation}
Denote by $R^i$ the random variable for the number of rounds needed for the message to be sent from $rt^{i-1}$ to $rt^i$, for each $i$,
and let $R$ be the random variable for the number of rounds needed to send the message from $rt$ to $rt^{D/2}$.
We thus have $R =  R^1 + R^2 + \ldots + R^{D/2}$.

By Corollary \ref{detcor}, the number of rounds needed for the broadcast message to reach half of the vertices of $N^1$ (from $rt = rt^0$) is at least
$((\log 2s) / 4) + 1 = \Theta(\log (n/D))$.
Since $rt^1$ was sampled randomly from all vertices $N^1$ and as none of the processors have any topological information on the graph,
$rt^1$ received this message within this many rounds with probability at most $1/2$,
hence $R^1 = \Omega(\log (n/D))$ with constant probability.
By Observation \ref{imm}, the only way for the message to reach any vertex of $S^2$, and later $rt^2$, is via $rt^1$, hence we can repeat this argument, and carry it out inductively.
Since the $D/2$ variables $R^1,R^2,\ldots,R^{D/2}$ are independently and identically distributed, and as $D = \Omega(\log n)$ (where the constant hiding in the $\Omega$-notation is sufficiently large), a   Chernoff bound
%in a straightforward way to
implies that $\Prob(R = \Omega(D \log(n/D))) \ge 1-n^{-c}$, where $c$ is a constant as big as needed.
% derive the required bound with high probability.
For the expectation bound,
note that $\Exp(R^i) > (\log 2s) / 4 = \Omega(\log (n/D))$ by Corollary \ref{detcor}, for each $i$,
and by linearity of expectation we obtain $\Exp(R) =  \Exp(R^1) + \Exp(R^2) + \ldots + \Exp(R^{D/2}) = \Omega(D \log(n/D))$.
(The assumption that $D = \Omega(\log n)$ is used for deriving the high probability bound but not the expectation bound.)

\subsubsection*{Acknowledgments.}
We are grateful to Mohsen Ghaffari for the useful discussions on the probabilistic arguments of Section \ref{sec:positive}.

\def\APPENDDETERMENSTIC{
\subsection{Bounds depending on the maximum degree}
\subsubsection{A naive approach}
In this section we provide a simple argument showing that when the maximum degree is small, the wireless expansion $\beta_w$ is not much smaller than
the ordinary expansion $\beta$. Recall that we consider an arbitrary bipartite graph $G_S = (S, \NN, E_S)$ with sides $S$ and $\NN$, such that $|\NN| = \beta \cdot |S|$.
We assume that no vertex of $G_S$ is isolated, i.e., all vertex degrees are at least 1.
In what follows we define $s=|S|$, $\gamma= |N|$.
\begin{lemma} \label{lem:naive}
In $G_S=(S,\NN,E_S)$, if the maximum degree is $\Delta$,
then there is a subset $S'$ of $S$ with $|\UN_{S}(S')| \ge \gamma/\Delta$.
\end{lemma}
  %Proving that $\beta_w \ge \beta/d^2$ is much easier.
%\Proof
%Let $S$ be a vertex set with $|S| \le \alpha n$, and denote by $\Gamma_{\bar S}(S)$ the set of their neighbors outside $S$, with $|\Gamma_{\bar S}(S)| \ge \beta \cdot |S|$.
\begin{proof}
We describe a procedure for computing vertex sets  $ \SETu \subseteq S$ and $\NNu \subseteq \NN$,
such that $|\NNu| \ge \gamma/\Delta$ and every vertex of $\NNu$ has a unique neighbor in $\SETu$.

Initialize $\NNu =  \SETu = \emptyset, \NNt = \NN,  \SETt = S$.
At each step of the procedure, the sets $\NNu$ and $\SETu$ (respectively, $\NNt$ and $\SETt$) grow (resp., shrink).
The procedure maintains the following invariant throughout.
\Invariant \begin{description}
\item{(I1)} $\SETt \cup \SETu \subseteq S$ and $\SETt \cap \SETu = \emptyset$.
\item{(I2)} $\NNt \cup \NNu \subseteq \NN$
and $\NNt \cap \NNu = \emptyset$.
\item{(I3)} Every vertex of $\NNu$ has a unique neighbor in $\SETu$.
\item{(I4)} Every vertex of $\NNt$ has at least one neighbor in $\SETt$, but has no neighbor in $\SETu$.
\end{description}

For a vertex $x \in \NNt$, recall that $\Gamma(x,{\SETt})$ is the set of neighbors of $x$ in $\SETt$.
At each step we pick a vertex $v \in \NNt$ minimizing
$|\Gamma(v,{\SETt})|$, i.e., a vertex with a minimum number of neighbors in $\SETt$.
%[[S: need to argue that this measure is not zero! That is, that each vertex in $\Gamma$ is covered by at least one vertex in $S$.
%And this is why we need the issue of removing all vertices of $Q''_v$ incident on $w$ below]]
%(The vertex set $S^{temp}$ will be shrinking at each step, so this measure $|\Gamma_{S^{temp}}(v)|$ is dynamic.
(By invariant $(I4)$, we have $|\Gamma(v,{\SETt})| \ge 1$.)
Let $Q_v$ be the set of all vertices in $\NNt$ that are incident on at least one vertex of $\Gamma(v,{\SETt})$.
By the choice of $v$, for any vertex $u$ in $Q_v$ satisfying $\Gamma(u,{\SETt}) \subseteq \Gamma(v,{\SETt})$, we must have $\Gamma(u,{\SETt}) = \Gamma(v,{\SETt})$.
%and consider an arbitrary vertex $q
We partition $Q_v$ into two subsets $Q'_v$ and $Q''_v$, where $Q'_v$ contains all vertices $u$ for which $\Gamma(u,{\SETt}) = \Gamma(v,{\SETt})$
and $Q''_v$ contains the remaining vertices of $Q_v$ (all of which must have a neighbor in $\SETt \setminus \Gamma(v,{\SETt})$).
Obviously we have $Q'_v \supseteq \{v\}$, so $|Q'_v| \ge 1$.
%, i.e., the vertices $v''$ for which $\Gamma_S(v'') \cap \Gamma_S(v) \ne \emptyset$
%and $\Gamma_S(v'') \setminus \Gamma_S(v) \ne \emptyset$.

We start by moving an arbitrary vertex $w$ of $\Gamma(v,{\SETt})$ from $\SETt$ to $\SETu$; note that
$w$ is incident on all vertices of $Q'_v$.
Then we remove all other vertices of $\Gamma(v,{\SETt})$ from $\SETt$, which prevents these vertices from entering $\SETu$ later on, thus guaranteeing that all vertices in $Q'_v$ will have $w$ as their unique neighbor in $\SETu$.
Subsequently,  all vertices of $Q'_v$ are moved from $\NNt$ to $\NNu$.
(See Figure \ref{fig:naivearg1} for an illustration.)

\def\FIGA{
%%%%%%%%%%%%%%%%
\begin{figure}[htb]
\begin{center}
\includegraphics[scale=0.6]{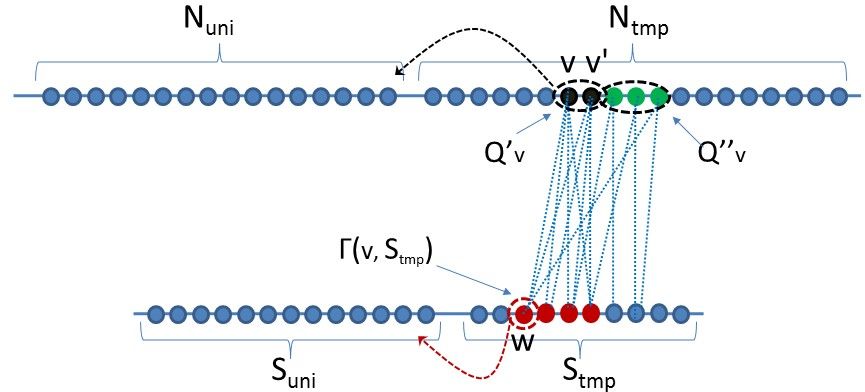}
\end{center}
\caption{\sf An illustration of a single step of the procedure.
%the vertex sets $\SETu,\SETt$ and $\NNu,\NNt$ into which $S$ and $\Gamma$ are partitioned,
The dashed lines represent edges that connect vertices in $Q_v$ with vertices in $\SETt$, where $v$ is a vertex in $\NNt$ minimizing $|\Gamma(v,{\SETt})|$.
The vertices in $Q'_v$ are colored black, and they move from $\NNt$ to $\NNu$; the vertices in $Q''_v$ are colored green, and they
are removed from $\NNt$; the vertices in $\Gamma(v,{\SETt})$ are colored red, and they are removed from $\SETt$, except for $w$ which moves to $\SETu$.}
\label{fig:naivearg1}
\end{figure}
%%%%%%%%%%%%%%%
}%FIGA
\FIGA

%It is easy to see that the third assertion of the invariant continues to hold.
In addition, to prevent violating invariant $(I4)$ now and invariant $(I3)$ in the future, all neighbors of $w$ that belong to $Q''_v$ are removed from $\NNt$  (they are incident to $w$ which has just moved to $\SETu$, and they might have neighbors in $\SETt$
that will be moved to $\SETu$ later on).
%this is exactly our loss, which is bounded by $d-1$, and so it is greater than our gain $|Q'_v| \ge 1$ by a factor of at most $d-1$.
It is clear that the first three invariants $(I1)-(I3)$ continue to hold following this step.
As for invariant $(I4)$, consider an arbitrary vertex $u$ of $\NNt$ at the beginning of this step.
We know that $u$ had no neighbors in $\SETu$ at the beginning of the step. If $u$ is not a neighbor of $w$, then $u$ had no neighbors in $\SETu$ also at the end of the step. Otherwise, If $u$ is a neighbor of $w$, then its only new neighbor in $\SETu$ at the end of the step is $w$ and $u$ was removed from $\NNt$ (it either moves to $\NNu$ if it belongs to $Q'_v$, or it is removed altogether if it belongs to $Q''_v$). This shows that every vertex $u$ of $\NNt$ has no neighbor in $\SETu$ at the end of the step.
Next, if $u$ has a neighbor outside $\Gamma(v,{\SETt})$, then this neighbor remains in $\SETt$ following the step
(since only the vertices of $\Gamma(v,{\SETt})$ are removed from $\SETt$ during the step).
Otherwise, we have $\Gamma(u,{\SETt}) \subseteq \Gamma(v,{\SETt})$, which by the choice of $v$ implies that $\Gamma(u,{\SETt}) = \Gamma(v,{\SETt})$.
By definition, $u \in Q'_v$, and is thus removed from $\NNt$ during the step. This shows that at the end of this step, every vertex of $\NNt$ has at least one neighbor in $\SETt$, so $(I4)$ holds.

This procedure terminates once $\NNt = \emptyset$.
By invariant $(I3)$, every vertex of $\NNu$ has a unique neighbor in $\SETu$.
At each step of the procedure, we move $|Q'_v| \ge 1$ vertices from $\NNt$ to $\NNu$, all of which are neighbors of some vertex
$w \in \Gamma(v,{\SETt})$, and remove some of the other (at most $\Delta-1$) neighbors of $w$ from $\NNt$.
Consequently, at least one vertex among every $\Delta$ vertices removed from $\NNt$ must move to $\NNu$.
%for every vertex moved from $\Gamma^{temp}$ in $\Gamma$, we removed at most $d-1$ other vertices from $\Gamma$,
Since initially we have $\NNt = \NN$,
it follows that $|\NNu| \ge \gamma/\Delta$.
\QED
\end{proof}

Note that the proof of this lemma takes into account the maximum degree $\Delta_S$ of a vertex in $S$,
rather than the maximum degree $\Delta$ in the entire graph.

\begin{corollary} \label{naive lower bound}
Suppose $G$ is an $(\alpha,\beta)$-expander with maximum degree $\Delta$. Then it is also an ($\alpha_w ,\beta_w$)-wireless expander, with $\alpha_w=\alpha$ and $\beta_w\geq\beta/\Delta$.
\end{corollary}
%Note that this naive bound is already better than that of Lemma \ref{from unique to wireless} when $\beta\leq \Delta^2/(2\Delta-1)$, and better than that of Corollary \ref{lower bound by probability} when $\beta\leq 8$.

\subsubsection{Procedure Partition} \label{basicproc}
Our next goal is to strengthen Corollary \ref{naive lower bound}. In this section we describe a procedure, hereafter named Procedure {\partition}, which lies at the core of our lower bounds on the wireless expansion. This procedure is then employed in various scenarios to conclude that the wireless expansion is close to the ordinary expansion. The procedure partitions $\NN$ into
$\NNu,\GAMm,\NNt$ and $S$ into $\SETu$ and $\SETt$,
such that the following conditions hold. (In what follows we refer to these conditions as the ``partition conditions''.)
\begin{description}
\item{(P1)} Every vertex of $\NNu$ has a unique neighbor in $\SETu$.
\item{(P2)} Every vertex of $\NNt$ has at least one neighbor in $\SETt$, but has no neighbor in $\SETu$.
\item{(P3)} $|\NNu| \ge |\GAMm|$.
\item{(P4)} Either $\NNt = \emptyset$, or $|\EEt|  \le 2|\EEu|$ holds,
where $\EEu$ (resp., $\EEt$)
denotes the set of edges connecting all vertices in $\SETt$ with vertices in $\NNu$ (resp., $\NNt$).
\end{description}
%\end{lemma}
%\begin{proof}

At the outset, we initialize $\NNu = \GAMm = \SETu = \emptyset, \NNt = \NN,  \SETt = S$.
At each step of the procedure, the sets $\NNu$ and $\SETu$ grow and the set $\NNt$ and $\SETt$ shrink.
The set $\GAMm$ also grows, but not necessarily at each step; it contains ``junk'' vertices that once belonged to $\NNu$,
but were removed from  $\NNu$ due to new vertices added to $\SETu$.

The first three aforementioned conditions are maintained throughout the execution of the procedure. (Notice that initially all three of them hold trivially.)
On the other hand, condition $(P4)$ is required to hold only when the procedure terminates.
%\begin{invariant} \label{basicinvariant}
%\begin{enumerate}
%\item $\SETt \cup \SETu= S$ and $\SETt \cap \SETu = \emptyset$.
%\item $\NNt \cup \NNu \cup \GAMm = \Gamma_{2\delta}$
%and $\NNt \cap \NNu = \NNt \cap \GAMm = \NNu \cap \GAMm = \emptyset$.
%\item Every vertex of $\NNu$ has a unique neighbor in $\SETu$.
%\item Every vertex of $\NNt$ has at least one neighbor in $\SETt$, but has no neighbor in $\SETu$.
%\item $|\NNu| \ge |\GAMm|$.
%\end{enumerate}
%\end{invariant}

For a vertex $x \in \SETt$, denote by  $\NNt(x)$ (resp., $\NNu(x)$) the set of neighbors of $x$ in $\NNt$ (resp., $\NNu$).
%note that all vertices in $\Gamma^{temp}(x)$ are not adjacent to any vertex in $\SETu$.
%Similarly, denote by  the set of neighbors of $x$ in .  %; note that each vertex in $\NNu$
%is already adjacent to at least one vertex of $\SETu$.

At each step we pick a vertex $v \in \SETt$ maximizing $gain(v) := |\NNt(v)| - 2|\NNu(v)|$.
% roughly speaking,
%this vertex has a maximal number of vertices that haven't been discovered yet.
%[[S: need to argue that this measure is no zero! That is, that each vertex in $\Gamma$ is covered by at least one vertex in $S$.
%And this is why we need the issue of removing all vertices of $Q''_v$ incident on $w$ below]]
%(The vertex set $\hat S$ will be shrinking at each step, so this measure $|\Gamma_{\hat S}(v)|$ is dynamic.
%(By the fourth assertion of the invariant, we have $|\Gamma_{\hat S}(v)| \ge 1$.)
Assuming $gain(v) > 0$, we move $v$ from $\SETt$ to $\SETu$;
to preserve condition $(P1)$, we move the vertices of $\NNu(v)$ from $\NNu$ to $\GAMm$.
Next, we move all vertices of $\NNt(v)$ from $\NNt$ to $\NNu$. Since $gain(v) > 0$, condition $(P3)$ holds.
The reason condition $(P2)$ holds is because once a vertex of $\SETt$ moves to $\SETu$, all its neighbors in $\NNt$
  are moved to $\NNu$. Obviously the sets $\NNu,\GAMm,\NNt$ (resp., $\SETu,\SETt$) form a partition of $\NN$ (resp., $S$).

Procedure {\partition} terminates once $\SETt$ becomes empty or once $gain(v) \le 0$ for all $v \in \SETt$.
In the former case, condition $(P2)$ implies that $\NNt = \emptyset$, and we are done. In the latter case, we have $|\NNt(v)| \le 2|\NNu(v)|$ for any $v \in \SETt$, yielding
\begin{equation} \label{secondcase}
|\EEt| ~=~ \sum_{v \in \SETt} |\NNt(v)| ~\le~   \sum_{v \in \SETt} 2|\NNu(v)| ~=~ 2|\EEu|.
\end{equation}
%Denote by $\EEu$ the set of edges connecting all vertices in $\SETt$ with vertices in $\NNu$,
%and denote by $\EEt$ the set of edges connecting all vertices in $\SETt$ with vertices in $\NNt$.
%the vertex sets $\SETu,\SETt$ and $\NNu,\GAMm,\NNt$ into which $S$ and $\Gamma_{2\delta}$ are partitioned, and

%\QED
%\end{proof}

(See Figure \ref{fig:mainarg1} for an illustration.)

\def\FIGB{
%%%%%%%%%%%%%%%%
\begin{figure}[htb]
\begin{center}
\includegraphics[scale=0.6]{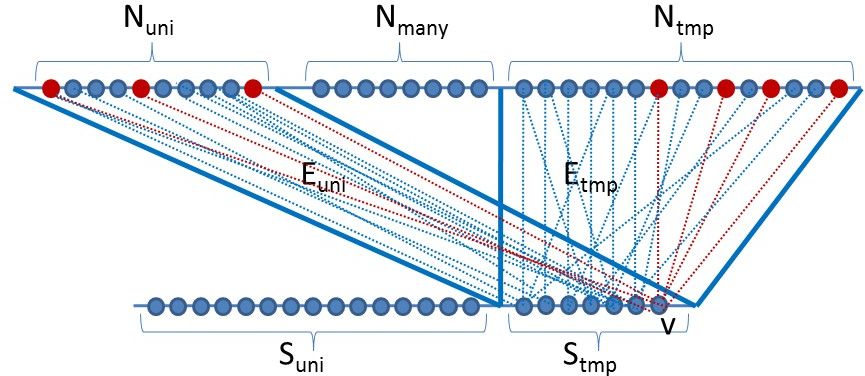}
\end{center}
\caption{\sf An illustration of
%the vertex sets $\SETu,\SETt$ and $\NNu,\GAMm,\NNt$ into which $S$ and $\Gamma_{2\delta}$ are partitioned, and
the edge sets $\EEu$ and $\EEt$ which connect $\SETt$ to $\NNu$ and $\NNt$, respectively. The vertices in $\NNt(v)$ and $\NNu(v)$, as well as the edges connecting them to $v$, are colored red; here we have $\gain(v) = |\NNt(v)| - 2|\NNu(v)| = -2$.}
\label{fig:mainarg1}
\end{figure}
%%%%%%%%%%%%%%%%
}%FIGB
\FIGB

\subsubsection{Constructive lower bound for $\beta_w$ in terms of the average degree}
Let  $\NN=\EN$ and $\gamma=|\NN|$, and denote by $\delta$ the average degree of a vertex in $\NN$, i.e., $\delta = (1/\gamma) \sum_{v \in \NN} \deg(v,S)$.
(As all vertex degrees are at least 1, we have $\delta \ge 1$.)\\
We next show a lower bound that takes into account the average degree $\delta$ rather than the maximum degree $\Delta$.

\begin{lemma}
\label{constructive lemma}
In the graph $G_S$ there exists a subset $S'$ of $S$ with $|\UN_S(S')|\geq \gamma/(8\delta)$.
\end{lemma}
\begin{proof}
Denote by $\NN^{2\delta}$ the set of vertices of $\NN=\EN$ with degree at most
$2\delta$. Observe that at least half the vertices of $\NN$ have degree
at most twice the average, implying that $|\NN^{2\delta}| \ge \gamma/2$.
%\footnote{
%(We assume each vertex in $\Gamma_{\bar S}(S)$ has at most $2d/\beta$ neighbors in $S$; this assumption may effect the final upper bound by a factor of 2.}
%(The set $\bar S$ denotes the complementary vertex set of $S$ in $G$, i.e., $\bar S = V \setminus S$.)
We apply Procedure {\partition}, but consider the vertex set $\NN^{2\delta}$
rather than $\NN$. Thus we obtain a partition of $\NN^{2\delta}$ rather than
$\NN$ into $\NNu^{2\delta},\GAMm^{2\delta},\NNt^{2\delta}$ and a partition of $S$
into $\SETu$ and $\SETt$ satisfying the partition conditions $(P1)-(P4)$.
Next, we show that $|\NNu^{2\delta}| \ge  \gamma/(8\delta)$.

%Note that the sum of degrees of vertices in $S$ is bounded by $S \cdot d$.
%As a result, the number of vertices in $\Gamma_{\bar S}(S)$
%We may assume that all vertices in $\Gamma_{\bar S}(S)$ have degree bounded by $2d/\beta$;
%To complete the proof, we show that upon termination  it holds that $|\NNu| \ge  |\Gamma|/(8\delta)$.

Suppose first that the procedure terminates because $\NNt^{2\delta} = \emptyset$.
By partition condition $(P3)$, $2|\NNu^{2\delta}| \ge |\NNu^{2\delta}| + |\GAMm^{2\delta}| = |\NN^{2\delta}|$.
It follows that
\begin{equation} \label{firstcase}
|\NNu^{2\delta}| ~\ge~  \frac{|\NN^{2\delta}|}{2} ~\ge~ \frac{\gamma}{4} ~\ge~ \frac{\gamma}{4\delta}.
\end{equation}
We henceforth assume that $|\EEt|  \le 2|\EEu|$.
By definition, each vertex in $\NN^{2\delta}$ has at most $2\delta$ neighbors in $S$. Condition $(P1)$ implies that each vertex in $\NNu^{2\delta}$ has a
single neighbor in $\SETu$, so it has at most $2\delta-1$ neighbors in $\SETt$, yielding $|\EEu| ~\le~ (2\delta-1)  |\NNu^{2\delta}|.$
By condition $(P2)$, each vertex of $\NNt^{2\delta}$ is incident on at least one edge of $\EEt$, and so $|\EEt| \ge |\NNt^{2\delta}|$.
It follows that
\begin{eqnarray*}
|\NNt^{2\delta}| ~\le~ |\EEt| ~\le~ 2|\EEu| ~\le~ (4\delta -2)   |\NNu^{2\delta}|.
\end{eqnarray*}
Hence,
\begin{eqnarray*}
4\delta \cdot |\NNu^{2\delta}| &=& (2 + (4\delta -2))  |\NNu^{2\delta}|
\\&\ge& |\NNu^{2\delta}| + |\GAMm^{2\delta}| + |\NNt^{2\delta}|
\\&=& |\NN^{2\delta}|  ~\ge~ \frac{\gamma}{2},
\end{eqnarray*}
which yields $|\NNu^{2\delta}| ~\ge~ \gamma/(8\delta).$ \QED
\end{proof}
%{\bf Remark.} We consider the average degree of a vertex in $\NN=\EN$ rather than in $S$. As we will show later, this is for a good reason --
%there are examples where the average degree of a vertex in $S$ is constant yet one cannot uniquely cover more than $O(\gamma / \log\gamma)$ vertices of $\EN$.

For every $S\subset V$ denote by $\delta_S$ the average degree of a vertex in $\NN=\EN$, i.e., $\delta_S=(1/|\NN|)\sum_{v\in \NN}\deg (v,S)$ and denote $\bar{\delta}=\max\{\delta_S~|~ S\subset V, |S|\leq \alpha n\}$.

\begin{corollary}
\label{basiclemma1}
Let $G=(V,E)$ be an $(\alpha, \beta)$-expander. Then \item{(1)} $G$ is an $(\alpha_w,\beta_w)$-wireless expander with $\alpha_w=\alpha$ and $\beta_w \geq \beta /(8\bar{\delta}) \ge \beta /(8\Delta)$, where $\Delta$ is the maximum degree in the graph. \item{(2)}
In the regime $\beta \ge 1$, we have $\delta_S \le \Delta / \beta$, for every $S$ such that $|S|\leq \alpha n$, thus $\bar{\delta}\leq \Delta/\beta$ and we get $\beta_w \ge \beta^2 /(8\Delta)$.
\end{corollary}
\def\APPENDF{
%\Proof
For every set $S\subset V$, such that $|S|\leq \alpha n$, let $\NN=\EN$ and consider the corresponding bipartite graph $G_S=(S,\NN,E_S)$. By Lemma \ref{constructive lemma}, there exists a subset $S'$ of $S$ with $\UN_S(S')\geq \gamma/(8\delta)\geq \gamma/(8\bar{\delta})$. By the expansion property, $|\gamma|=|\EN|\geq \beta |S|$. Hence, $|\UN_S(S')|\geq \beta/(8\bar{\delta})$.
\inQED
} %APPENDF

\subsubsection{``Convenient'' degree constraints}
The following lemmas show that if many vertices in $\EN$ have roughly the same degree, then the ordinary expansion $\beta$ and the wireless expansion $\beta_w$ of $G_S$ are roughly the same.

\begin{lemma}
\label{partition using the degrees}
In $G_S$, for any   $c > 1$ and any $i \in \{1,2,\ldots,\log_c |S|\}$,
%Denote by $\delta$ the average degree of a vertex in $\Gamma$, i.e., $\delta = (1/|\Gamma|) \sum_{v \in \Gamma} \deg(v)$.
%(As all vertex degrees are at least 1, we have $\delta \ge 1$.)
%If the average degree in $G_S$ is $d$,
there is a subset $S'$ of $S$ with $|\UN_{S}(S')| \ge |\NN^{(i)}|/2(1+c)$,
where $\NN^{(i)}$ denotes the set of vertices in $\NN$ with degree in $[c^{i-1},c^i)$ for $i<\log_c|S|$ and for $i=\log_c|S|$ is the set of vertices in $\NN$ with degree in $[c^{i-1},c^i]=[|S|/c,|S|]$.
\end{lemma}
\def\APPENDH{
%\Proof
We apply Procedure $\partition$, but consider the vertex set $\NN^{(i)}$ rather than $\NN=\EN$. Thus we obtain a partition of $\NN^{(i)}$ rather than $\NN$ into $\NNu^{(i)},\GAMm^{(i)},\NNt^{(i)}$
and a partition of $S$ into $\SETu$ and $\SETt$ satisfying the partition conditions $(P1)-(P4)$.
Next, we show that $|\NNu^{(i)}| \ge  |\NN^{(i)}|/2(1+c)$. This complete the proof as $|\UN_{S}(\SETu)|\geq|\UN_{S}(\SETu)\cap \NN^{(i)}|= |\NNu^{(i)}|$.

If the procedure terminates because $\NNt^{(i)} = \emptyset$,
then we have $|\NNu^{(i)}|~\ge~  |\NN^{(i)}|/2$ (cf.\ Equation (\ref{firstcase})).

We henceforth assume that $|\EEt|  \le 2|\EEu|$.
Note that each vertex in $\NN^{(i)}$ has at most $\lfloor c^i\rfloor$ neighbors in $S$. Condition $(P1)$ implies that each vertex in $\NNu^{(i)}$ has a
single neighbor in $\SETu$, so it has at most $\lfloor c^i\rfloor-1$ neighbors in $\SETt$,
yielding $|\EEu| ~\le~ (\lfloor c^i\rfloor-1)  |\NNu^{(i)}|.$
By condition $(P2)$, each vertex of $\NNt^{(i)}$ is incident only on edges of $\EEt$.
Since there must be at least $\lceil c^{i-1}\rceil$ such edges, we have $|\EEt| \ge \lceil c^{i-1}\rceil|\NNt^{(i)}|$.
It follows that $$\lceil c^{i-1}\rceil\cdot |\NNt^{(i)}| ~\le~ |\EEt| ~\le~ 2|\EEu| ~\le~ 2(\lfloor c^i\rfloor -1)\cdot   |\NNu^{(i)}|.$$
Hence,
\begin{eqnarray*}
2(c^{i-1} + c^i)  |\NNu^{(i)}| &=& (2(c^{i-1} +1) + 2(c^i-1))  |\NNu^{(i)}|
\ge(2(\lceil c^{i-1} \rceil) + 2(\lfloor c^i\rfloor-1))  |\NNu^{(i)}| \\ &\ge& \lceil c^{i-1} \rceil(|\NNu^{(i)}| + |\GAMm^{(i)}| + |\NNt^{(i)}|)
  = \lceil c^{i-1} \rceil |\NN^{(i)}|,
\end{eqnarray*}
which yields $$|\NNu^{(i)}| ~\ge~ \frac{\lceil c^{i-1} \rceil}{2(c^{i-1} + c^i)} |\NN^{(i)}|  ~\ge~ \frac{c^{i-1}}{2(c^{i-1} + c^i)}|\NN^{(i)}|  ~\ge~ \frac{ |\NN^{(i)}|}{2(1+ c)}~. \inQED$$
}%APPENDH

\begin{corollary} \label{samedegrees}
In $G_S$, for any   $c > 1$ there is a subset $S'$ of $S$ such that $$|\UN_{S}(S')| \ge  \frac{\log_2c}{2(1+c)\log_2\Delta}\cdot\gamma~.$$
\end{corollary}
\begin{proof}
The previous lemma implies also that for every $c > 1$ and every
$i \in \{1,2,\ldots,\log_c \Delta \}$ (rather then $\log_c|S|$),
there is a subset $S'$ of $S$ with $|\UN_{S}(S')| \ge |\NN^{(i)}|/2(1+c)$,
where $\NN^{(i)}$ denotes the set of vertices in $\NN$ with degree in
$[c^{i-1},c^i)$ for $i<\log_c\Delta$ and for $i=\log_c\Delta$ is the set of
vertices in $\NN$ with degree in $[c^{i-1},c^i]=[\Delta/c,\Delta]$.
Observe that there exists an index $j \in \{1,2,\ldots,\log_c \Delta \}$
s.t. $|\NN^{(j)}|\geq \gamma/\log_c{\Delta}$,  for this index $j$, we get
$|\UN_{S}(S')| \ge |\NN^{(j)}|/2(1+c)\geq \frac{\gamma}{ 2(1+c)\log_c{\Delta}}$.
\QED
\end{proof}
%}%APPENDI

%Denote by $\delta$ the average degree of a vertex in $\Gamma$, i.e., $\delta = (1/|\Gamma|) \sum_{v \in \Gamma} \deg(v)$.
%(As all vertex degrees are at least 1, we have $\delta \ge 1$.)
%If the average degree in $G_S$ is $d$,

The maximum of $f(c)=\log_2c/(2(1+c))$ is attained at $c\approx 3.59112$
and equals $\approx 0.20087$, hence we get the following.

\commabs
\begin{corollary} \label{samedegrees1}
In $G_S$, there is a subset $S'$ of $S$ such that
$\displaystyle |\UN_{S}(S')| \ge  \frac{0.20087}{\log_2\Delta}\cdot\gamma~.$
\end{corollary}
\commabsend

\begin{corollary}\label{samedegrees1}
Let $G=(V,E)$ be an $(\alpha,\beta)$-expander with maximum degree $\Delta$. Then it is also an $(\alpha_w,\beta_w)$-wireless expander with $\alpha_w\geq \alpha$ and $\beta_w\geq \frac{0.20087}{\log_2 {\Delta}}\cdot\beta$.
\end{corollary}

\commabs
\begin{proof} Given $S\subset V$ with $|S|\leq \alpha|V|$, write $\gamma=|\EN|$ and let $G_S=(S,\EN, e(S,\EN)$.
Note that as $G$ is an $(\alpha,\beta)$-expander, $\gamma\geq\beta|S|$ and by Corollary \ref{samedegrees1}
there is a subset $S'$ of $S$ with $|\UN_{S}(S')| \geq  \frac{0.20087}{\log_2\Delta}\cdot\gamma\geq  \frac{0.20087}{\log_2\Delta}\cdot\beta|S|$, and therefore $\beta_w\geq\frac{0.20087}{\log_2\Delta} \cdot\beta $~. 
\QED
\end{proof}
\commabsend

%%%%%%%%%%%%%%%%
\subsection{Bounds depending on the average degree}\label{Bounds depending on the average degree}
Recall that $\delta$ denotes the average degree of a vertex in $\NN=\EN$. In case $\delta$ is known, we can state a stronger bound than that of Corollary \ref{samedegrees1}, using $\delta$ in place of $\Delta$.

\begin{corollary}
\label{lower bound on gamma1 with t and c}
In $G_S$, for any   $c > 1$ and $t>1$ there is a subset $S'$ of $S$ such that $$|\UN_{S}(S')| \ge \left(1-\frac{1}{t}\right) \frac{1}{2(1+c)\log_c(t\delta)}\cdot\gamma.$$
\end{corollary}
\def\APPENDJ{
%\Proof
Lemma \ref{partition using the degrees} implies also that for any   $c,t > 1$ and any $i \in \{1,2,\ldots,\log_c(t\delta)\}$, there is a subset $S'$ of $S$ with $|\UN_{S}(S')| \ge |\NN^{(i)}|/2(1+c)$,
where $\NN^{(i)}$ denotes the set of vertices in $\NN$ with degree in $[c^{i-1},c^i)$ for $i<\log_c(t\delta)$ and for $i=\log_c(t\delta)$ is the set of vertices in $\NN$ with degree in $[c^{i-1},c^i]=[t\delta/c,t\delta]$. Observe that there exists an index $j \in \{1,2,\ldots,\log_c (t\delta) \}$ s.t. $|\NN^{(j)}|\geq |\NN^{t\delta}|/\log_c(t\delta)\geq |\NN|(1-1/t)(1/\log_c(t\delta))=\gamma(1-1/t)(1/\log_c(t\delta))$, where $\NN^{t\delta}$ is the set of vertices with degree at most $t\delta$. For this index $j$, Lemma \ref{partition using the degrees} implies $|\UN_{S}(S')| \ge |\NN^{(j)}|/2(1+c)\geq |\NN^{t\delta}|/(2(1+c)\log_c(t\delta))\geq \gamma(1-1/t)(1/(2(1+c)\log_c(t\delta)))$ ~. \inQED
} %APPENDJ

\begin{corollary}
\label{finding t and c}
In $G_S$, for every $\epsilon>0$, and for sufficiently large\footnote{$\delta$ that that satisfies $\epsilon\ln(\delta)-\ln(\ln\delta)-\ln(1+\epsilon)-1\geq 0$ is enough.} $\delta$, there is a subset $S'$ of $S$ such that $|\UN_{S}(S')| \ge \frac{2.0087}{(1+\epsilon)\log_2(\delta)}\cdot\gamma.$
\end{corollary}
\def\APPENDK{
%\Proof
Let $\epsilon>0$ and let $f(c)= \log_2(c)/(2(1+c))$ and
$g(t)=(1-1/t)(1/\log_2(t\delta))$. By corollary
\ref{lower bound on gamma1 with t and c}, for any $c > 1$ and $t>1$ there is
a subset $S'$ of $S$ such that $|\UN_{S}(S')| \ge \gamma \cdot f(c) g(t)$.
Note that $f(c)$ and $g(t)$ are both positive in $c>1$, $t>1$, thus,
\begin{eqnarray}\label{maximum of f and g}
|\UN_{S}(S')| &\ge& \gamma \cdot\max\{ f(c)\cdot g(t)~|~ c>1, t>1\}\nonumber
\\&=&
\gamma \cdot\max\{\cdot f(c)~|~c>1\}\cdot\max\{ g(t)~|~t>1\},
\end{eqnarray}
and it is enough to find $\max \{f(c)~|~ c>1\}$ and $\max \{g(t)~|~ t>1\}$.
Inspecting the derivative $f'(c)=(1/\ln2)\frac{2(1+c)/c-2\ln c}{(2+2c)^2}$
reveals that $f'(c)=0$ if and only if $1+1/c=\ln c$, and the maximum
is attained at $c^{max}\approx 3.59112$ and equals $f(c^{max})\approx 0.20087$.
On the other hand,
$g'(t)=\ln2\left((1/t^2)(1/\ln(t\delta))-(1-1/t)(1/(t\ln^2(t\delta)))\right)
=(\ln2/(t^2\ln^2(t\delta))(\ln(t\delta)-t+1)$, and $g'(t)=0$ if and only if
$h(t)=\ln(t\delta)-t+1=0$.
\\ {\bf DP: What is $h$ here?} \\
Note also that $g'(t)\geq 0$ when $h(t)\geq0$ and $g'(t)\leq 0$ when $h(t)\leq0$. As $h(t)$ is a continuous decreasing function, $h(t)=0$ yields a maximum point of $g(t)$. Plugging $t=\ln\delta$ we get $h(\ln\delta)<0$,  and by the assumption, $h((1+\epsilon)\ln\delta)=\epsilon\ln\delta-\ln(\ln\delta)-\ln(1+\epsilon)-1\geq 0$. Therefore, the maximum of $g(t)$ is attained at $t^{max}\leq (1+\epsilon)\ln\delta$, and $\ln(t^{max}\delta)-t^{max}-1=0$. Hence, in that point, $g(t^{max})=\ln2/t^{max}\geq \ln2/((1+\epsilon)\ln \delta)=1/((1+\epsilon)\log_2\delta)$. In conclusion, plugging $f(c^{max})$ and $g(t^{max})$ in Eq. (\ref{maximum of f and g}), we get
\begin{eqnarray*}
|\UN_{S}(S')| &\ge& \gamma \cdot f(c^{max})\cdot g(t^{max})
\\&\geq& \gamma\cdot \frac{2.0087}{(1+\epsilon)\log_2(\delta)} ~.\inQED
\end{eqnarray*}
} %APPENDK

%For every $S\subset V$ denote by $\delta_S$ the average degree of a vertex in $\NN=\EN$, i.e., $\delta_S=(1/|\NN|)\sum_{v\in \NN}\deg (v,S)$ and denote $\bar{\delta}=\max\{\delta_S~|~ S\subset V, |S|\leq \alpha n\}$.

\begin{corollary}\label{lower bound in t, c and delta}
Let $G=(V,E)$ be an $(\alpha,\beta)$-expander with maximum degree $\Delta$ and let $\epsilon>0$. Suppose that for every $S$, $\delta_S$ is large\footnote{i.e., satisfies $\epsilon\ln(\delta_S)-\ln(\ln\delta_S)-\ln(1+\epsilon)-1\geq 0$.} enough. Then $G$ is also an $(\alpha_w,\beta_w)$-wireless expander with $\alpha_w\geq \alpha$ and $$\beta_w\geq \frac{2.0087}{(1+\epsilon)\log_2(\bar{\delta})}\cdot\beta~.$$
\end{corollary}
%\commabs
\begin{proof}
Given $S\subset V$ with $|S|\leq \alpha|V|$, write $\gamma=|\EN|$ and let $G_S=(S,\EN, e(S,\EN)$. Note that as $G$ is an $(\alpha,\beta)$-expander, $\gamma\geq\beta|S|$ and by Corollary \ref{finding t and c}, there is a subset $S'$ of $S$ with
\begin{eqnarray*}
|\UN_{S}(S')| &\ge&  \frac{2.0087}{(1+\epsilon)\log_2(\delta_S)}\cdot\gamma\\&\geq&  \frac{2.0087}{(1+\epsilon)\log_2(\delta_S)}\cdot\beta|S|\geq  \frac{2.0087}{(1+\epsilon)\log_2(\bar{\delta})}\cdot\beta|S|~.
\end{eqnarray*}
Hence $\beta_w\geq \frac{2.0087}{(1+\epsilon)\log_2(\bar{\delta})}\cdot\beta.$
\QED
\end{proof}
%\commabsend

\begin{lemma}
\label{lem:c and t}
Suppose there exists $c>1$ and $t>1$ such that for every subset $\NN'$ of $\NN$ in $G_S$ of sufficiently large size (say, of size at least $(\gamma/2)(1-1/t)$), the average degree $\delta'$ of a vertex in $\NN'$ is at least $t\delta/c$. Then there is a subset $S'$ of $S$ such that $$|\UN_{S}(S')| \ge \frac{\gamma}{2(1+ c)}\left(1-\frac{1}{t}\right).$$
\end{lemma}
\begin{proof}
We apply Procedure $\partition$, but consider the vertex set $\NN^{t\delta}$  of vertices in $\NN=\EN$ with degree at most $t\delta$. Thus we obtain a partition of $\NN^{t\delta}$ rather than $\NN$ into $\NNu^{t\delta},\GAMm^{t\delta},\NNt^{t\delta}$
and a partition of $S$ into $\SETu$ and $\SETt$ satisfying the partition conditions $(P1)-(P4)$.
Next, we show that $|\NNu^{t\delta}| \ge  |\NN^{t\delta}|/2(1+ c)$. This complete the proof as $|\NN^{t\delta}|\geq \gamma(1-1/t)$.

If $|\NNt^{t\delta}| < (\gamma/2)(1-1/t)$, as $|\NN^{t\delta}|\geq \gamma(1-1/t)$ and by using partition condition $(P3)$ we get $2|\NNu^{t\delta}|\geq |\NNu^{t\delta}|+|\GAMm^{t\delta}|\geq (\gamma/2)(1-1/t)$, hence $|\NNu^{t\delta}|\geq(\gamma/4)(1-1/t)\geq (\gamma/(2(1+ c)))(1-1/t)$. Otherwise, if $|\NNt^{t\delta}| \geq (\gamma/2)(1-1/t)$, in particular nonempty and it must hold that $|\EEt|  \le 2|\EEu|$.
By definition, each vertex in $\NN^{t\delta}$ has at most $t\delta$ neighbors in $S$. Condition $(P1)$ implies that each vertex in $\NNu^{t\delta}$ has a
single neighbor in $\SETu$, so it has at most $t\delta-1$ neighbors in $\SETt$,
yielding $|\EEu| ~\le~ (t\delta-1)  |\NNu^{t\delta}|.$
By condition $(P2)$, each vertex of $\NNt^{t\delta}$ is incident only on edges of $\EEt$.
Since $|\NNt^{t\delta}|\geq (\gamma/2)(1-1/t)$, the average degree in this set is at least $t\delta /c$. Therefore, $|\EEt| \ge (t\delta/c)|\NNt^{t\delta}|$.
It follows that $$\frac{t\delta}{c} |\NNt^{2\delta}| ~\le~ |\EEt| ~\le~ 2|\EEu| ~\le~ 2(t\delta-1)   |\NNu^{t\delta}|.$$
Hence \begin{eqnarray*}
2\left(\frac{t\delta}{c}+t\delta\right)  |\NNu^{t\delta}| &\geq& \left(2\cdot\frac{t\delta}{c} + 2(t\delta-1)\right)  |\NNu^{t\delta}|
\\ &\ge& \frac{t\delta}{c}\cdot(|\NNu^{t\delta}| + |\GAMm^{t\delta}| + |\NNt^{t\delta}|)
 \\ &=&\frac{t\delta}{c} |\NN^{t\delta}|,
\end{eqnarray*}
which yields $$|\NNu^{t\delta}|  ~\ge~ \frac{ |\NN^{t\delta}|}{2(1+ c)} . \inQED$$
\end{proof}

\begin{corollary}
Let $G=(V,E)$ be an $(\alpha,\beta)$-expander and suppose there exists $c>1$ and $t>1$ such that for every subset $S$ of $V$ of size $|S|\leq \alpha n$ and for every subset $M$ of $\EN$ of sufficiently large size (say, of size at least $(|\EN|/2)(1-1/t)$), the average degree $\delta'$ of a vertex in $M$ is at least $(t\delta_S)/c$. Then $G$ is also an $(\alpha_w,\beta_w)$-wireless expander with $\alpha_w\geq \alpha$ and $$\beta_w\geq \frac{\beta}{4(1+c)}\left(1-\frac{1}{t}\right).$$
\end{corollary}
\begin{proof}
The proof follows similar lines as those in the proof of Corollary \ref{lower bound in t, c and delta}.
\QED
\end{proof}

%\begin{Remark}
%We can prove the previous Lemmas using a probabilistic argument, similar to that by Merav, with a tweak:
%The probability of an element to be taken into $S'$ should be inverse-linear to the ``common degree'' (rather than the average degree).
%The disadvantage is that the constants will deteriorate, but the advantage is that such a subset $S'$ can be ``computed'' in a single communication round, assuming the relevant information
%about degrees in the network is known to the processors.
%\end{Remark}

%%%%%%%%%%%%%%%%%%%%%%%%%%%%%%%%%%%
\subsubsection{Near-optimal bounds}

\begin{lemma}
\label{tightlemma}
%Denote by $\delta$ the average degree of a vertex in $\Gamma$, i.e., $\delta = (1/|\Gamma|) \sum_{v \in \Gamma} \deg(v)$.
%(As all vertex degrees are at least 1, we have $\delta \ge 1$.)
%If the average degree in $G_S$ is $d$,
In $G_S$ there is a subset $S'$ of $S$ with $|\UN_{S}(S')| \ge \gamma /(9 \log (2\delta))$.
\end{lemma}
\begin{proof}
We prove the existence of vertex sets $\SETu \subseteq S$ and $\NNu \subseteq \NN=\EN$,
such that $|\NNu| \ge  \gamma/(9 \log (2\delta))$ and every vertex of $\NNu$ has a unique neighbor in $\SETu$.
The proof is by induction on $\gamma$, for all values of $\delta \ge 1$. (Since $\delta \ge 1$, we have $\log (2\delta) \ge 1$.)
\\\emph{Basis: $\gamma \le 9$.}  Let $v$ be an arbitrary vertex of $S$ with at least one neighbor in $\NN$,
let $\SETu = \{v\}$, and let $\NNu$  be the (non-empty) neighborhood of $v$.
We thus have $|\NNu| \ge 1 \ge \gamma/(9\log (2\delta))$.
\\\emph{Induction step: Assume the correctness of the statement for all smaller values of $\gamma$, and prove it for $\gamma$.}
We apply Procedure {\partition} (with the bipartite graph induced by the sets $S$ and $\NN$).
If the procedure terminates because $\NNt = \emptyset$,
then we have $|\NNu| ~\ge~  |\NN|/2$ (cf.\ Equation (\ref{firstcase})).

%We henceforth assume that the procedure terminates because $\gain(v) \le 0$, for all $v \in \SETt$.
We henceforth assume that $\NNt \ne \emptyset$, i.e., $\gamma' = |\NNt| \ge 1$.
In particular, it must hold that $|\EEt|  \le 2|\EEu|$.
Denote by $\delta'$ the average degree of a vertex in $\NNt$, counting only neighbors that belong to $\SETt$.
By partition condition $(P2)$, the entire neighborhood of $\NNt$ is contained in $\SETt$; confusing as it might be, we do {\em not} make use of this property here.
We do use, however, another property guaranteed by partition condition $(P2)$: Each vertex of $\NNt$ has at least one neighbor in $\SETt$, which implies that $\delta' \ge 1$, thus $\log (2\delta') \ge 1$. Since $\NNt$ is non-empty, it must hold that $|\EEt| \ge 1$.
Hence $|\EEu| \ge |\EEt| /2 \ge 1/2$, yielding $|\EEu| \ge 1$.
Consequently, we have $|\NNu| \ge 1$, which in turn yields $1 \le \gamma' \le \gamma-1$.

Suppose first that $\gamma'/ \log (2\delta') \ge \gamma / \log (2\delta)$.
By the induction hypothesis for $\gamma'$ (restricting ourselves to the subgraph of $G_S$ induced by the vertex sets $\SETt$ and $\NNt$),
we conclude that there is a subset $\tilde S$ of $\SETt$ with
$|\UN_{\SETt}(\tilde S) \cap \NNt| \ge \gamma'/(9 \log (2\delta'))$, yielding
$$ |\UN_{S}(\tilde S)| ~\ge~ |\UN_{\SETt}(\tilde S) \cap \NNt| ~\ge~  \frac{\gamma'}{9 \log (2\delta')} ~\ge~ \frac{\gamma}{9 \log (2\delta)}.$$

We may henceforth assume that \begin{equation} \label{bound_gamma1} \frac{\gamma'}{ \log (2\delta')} ~<~ \frac{\gamma}{\log (2\delta)}.\end{equation}
Observe that $|\EEu| + |\EEt| \le |E_S| = \delta \cdot \gamma$.
%$\gamma  = |\NNu| + |\GAMm| + \gamma'$ and
%and $|\NNu| \ge |\GAMm|$.
By definition, $|\EEt| = \delta' \cdot \gamma'$.
It follows that
$$3\delta' \cdot \gamma' ~=~ 3|\EEt| ~\le~  2(|\EEu| + |\EEt|) ~\le~ 2 \delta \cdot \gamma,$$
yielding
\begin{equation} \label{bound_logdelta1}
\log (2\delta') ~\le~ \log (2\delta) + \log\left(\frac{\gamma}{3}\right) - \log\left( \frac{\gamma'}{2}\right).
\end{equation}
Plugging Equation (\ref{bound_logdelta1}) into Equation (\ref{bound_gamma1}), we obtain
\begin{equation} \label{eq:prefinale}
\gamma' ~<~  \frac{\gamma}{\log (2\delta)} \left(\log (2\delta) + \log \left(\frac{\gamma}{3}\right) - \log \left(\frac{\gamma'}{2}\right)\right).
\end{equation}
We may assume that $|\NNu| < \gamma/9$, as otherwise $|\NNu| \ge \gamma/9 \ge \gamma/(9 \log (2\delta))$ and we are done.
By partition condition $(P3)$, $|\NNu|  \ge |\GAMm|$.
Hence $\gamma = |\NNu| + |\GAMm| + \gamma'  \le 2|\NNu| + \gamma'$,
yielding $(\gamma-\gamma')/2 \le |\NNu| < \gamma/9$. Hence $2/3 (\gamma/\gamma') \le 6/7$, which gives
$$\log \left(\frac{\gamma}{3}\right) - \log \left(\frac{\gamma'}{2}\right) ~\le~ \log \left(\frac{6}{7}\right) ~\le~ -\frac{2}{9}.$$
It follows that
\begin{multline} \label{eq:finale}
%\begin{aligned}
\frac{\gamma}{\log (2\delta)} \left(\log (2\delta) + \log \left(\frac{\gamma}{3}\right) - \log \left(\frac{\gamma'}{2}\right)\right)
\\ \le \frac{\gamma}{\log (2\delta)} \left(\log (2\delta) -\frac{2}{9}\right).
%\end{aligned}
\end{multline}
Plugging Equation (\ref{eq:finale}) into Equation (\ref{eq:prefinale}) gives
\begin{eqnarray*}
\gamma' &\le& \frac{\gamma}{\log (2\delta)} \left(\log (2\delta) -\frac{2}{9}\right) ~=~ \gamma - \frac{2\gamma}{9 \log (2\delta)} \\ &\le& 2|\NNu| + \gamma' - \frac{2\gamma}{9 \log (2\delta)},
\end{eqnarray*}
implying that $|\NNu| \ge \gamma/(9 \log (2\delta))$.
\QED
\end{proof}

\begin{corollary} \label{tightlemma1}
%Denote by $\delta$ the average degree of a vertex in $\Gamma$, i.e., $\delta = (1/|\Gamma|) \sum_{v \in \Gamma} \deg(v)$.
%(As all vertex degrees are at least 1, we have $\delta \ge 1$.)
%If the average degree in $G_S$ is $d$,
Let $G=(V,E)$ be an $(\alpha, \beta)$-expander.
%For every subset $S$ of $V$, let $\delta_S=(1/|\EN|)\sum_{v\in \EN}\deg (v,S)$ be the average degree of vertices in $\EN$ in $S$, and denote by $\bar{\delta}=\max\{\delta_S~|~ S\subset V, |S|\leq \alpha n\}$
Then,
\begin{description}
\item{(1)} $G$ is an $(\alpha_w,\beta_w)$-wireless expander with $\alpha_w=\alpha$ and $\beta_w \ge \beta /(9\log (2\bar{\delta})) \ge \beta /(9 \log (2\Delta))$,
where $\Delta$ is the maximum degree in the graph.
\item{(2)} In the regime $\beta \ge 1$, we have $\delta_S \le \Delta / \beta$, thus $\bar{\delta}\leq \Delta/\beta$, and hence $\beta_w \ge \beta /(9 \log (2\Delta / \beta))$.
\end{description}
\end{corollary}

\begin{corollary}
\label{improving shay's lower bound}
In $G_S$ there is a subset $S'$ of $S$ with $$|\UN_{S}(S')| \ge \min \left\{ \frac{\gamma}{9 \log \delta}~,~\frac{\gamma}{20}\right\}.$$
\end{corollary}
\begin{proof}
We prove that if $\delta<2$ then there is a subset $S'$ of $S$ with $|\UN_{S}(S')| \ge \gamma/20$ and if $\delta\geq 2$ then there is a subset $S'$ of $S$ with $|\UN_{S}(S')| \ge \gamma /(9 \log \delta)$. The proof is by induction on $\gamma$.
\\\emph{Basis: $\gamma \le 9$.} If $\delta<2$, by Lemma \ref{tightlemma} there is a subset $S'$ of $S$ with $|\UN_{S}(S')|\geq \gamma/(9\log(2\delta))> \gamma/18\geq \gamma/20$. For $\delta\geq 2$, let $v$ be an arbitrary vertex of $S$ with at least one neighbor in $\NN$,
let $S' = \{v\}$,
 then $|\Gamma(v)|=|\UN_{S}(S')| \ge 1 \ge \gamma/(9\log \delta)$.
\\\emph{Induction step: Assume the correctness of the statement for all smaller values of $\gamma$, and prove it for $\gamma$.}
If $\delta <2$, then the same proof holds as in the basis case. Let assume $\delta\geq 2$ and therefore $\log \delta\geq 1$.
We apply Procedure {\partition} (with the bipartite graph induced by the sets $S$ and $\NN$).
If the procedure terminates because $\NNt = \emptyset$,
then we have $|\NNu| ~\ge~  |\NN|/2$ (cf.\ Equation (\ref{firstcase})).

%We henceforth assume that the procedure terminates because $\gain(v) \le 0$, for all $v \in \SETt$.
We henceforth assume that $\NNt \ne \emptyset$, i.e., $\gamma' = |\NNt| \ge 1$.
In particular, it must hold that $|\EEt|  \le 2|\EEu|$.
Denote by $\delta'$ the average degree of a vertex in $\NNt$, counting only neighbors that belong to $\SETt$. By partition condition $(P2)$ each vertex of $\NNt$ has at least one neighbor in $\SETt$, which implies that $\delta' \ge 1$, thus $\log (2\delta') \ge 1$. Since $\NNt$ is non-empty, it must hold that $|\EEt| \ge 1$.
Hence $|\EEu| \ge |\EEt| /2 \ge 1/2$, yielding $|\EEu| \ge 1$.
Consequently, we have $|\NNu| \ge 1$, which in turn yields $1 \le \gamma' \le \gamma-1$.

There are two cases. The first case is when $\delta'<2$. By Lemma \ref{tightlemma}, there is a subset $S'$ of $S$ s.t.
\begin{eqnarray} \label{gamma'/18}
|\UN_{S}(S')|\geq|\UN_{\SETt}(S')\cap \NNt| \ge \frac{\gamma'}{9 \log (2\delta')}\geq \frac{\gamma'}{18}.
\end{eqnarray}
 If $\gamma'<\gamma(9/10)$, then as $\gamma = |\NNu| + |\GAMm| + \gamma'  \le 2|\NNu| + \gamma'$, we get $|\NNu|\geq \gamma/20$. So we can assume $\gamma'\geq \gamma(9/10)$, and by Equation (\ref{gamma'/18}),
 $$|\UN_{S}(S')|\geq \frac{\gamma'}{18}\geq \frac{\gamma}{20}.$$
The second case is when $\delta'\geq 2$ and therefore $\log \delta\geq 1$.

Suppose first that $\gamma'/ \log \delta' \ge \gamma / \log \delta$.
By the induction hypothesis for $\gamma'$ (restricting ourselves to the subgraph of $G_S$ induced by the vertex sets $\SETt$ and $\NNt$),
we conclude that there is a subset $\tilde S$ of $\SETt$ with
$|\UN_{\SETt}(\tilde S) \cap \NNt| \ge \gamma'/(9 \log \delta')$, yielding
$$ |\UN_{S}(\tilde S)| ~\ge~ |\UN_{\SETt}(\tilde S) \cap \NNt| ~\ge~  \frac{\gamma'}{9 \log \delta'} ~\ge~ \frac{\gamma}{9 \log \delta}.$$

We may henceforth assume that \begin{equation} \label{bound_gamma} \frac{\gamma'}{\log \delta'} ~<~ \frac{\gamma}{\log\delta}.\end{equation}
Observe that $|\EEu| + |\EEt| \le |E_S| = \delta \cdot \gamma$.
%$\gamma  = |\NNu| + |\GAMm| + \gamma'$ and
%and $|\NNu| \ge |\GAMm|$.
By definition, $|\EEt| = \delta' \cdot \gamma'$.
It follows that
$$3\delta' \cdot \gamma' ~=~ 3|\EEt| ~\le~  2(|\EEu| + |\EEt|) ~\le~ 2 \delta \cdot \gamma,$$
yielding
\begin{equation} \label{bound_logdelta}
\log \delta' ~\le~ \log \delta + \log \left(\frac{\gamma}{3}\right) - \log \left(\frac{\gamma'}{2}\right).
\end{equation}
Plugging Equation (\ref{bound_logdelta}) into Equation (\ref{bound_gamma}), we obtain
\begin{equation} \label{eq:prefinale1}
\gamma' ~<~  \frac{\gamma}{\log \delta} \left(\log \delta + \log \left(\frac{\gamma}{3}\right) - \log \left(\frac{\gamma'}{2}\right)\right).
\end{equation}
We may assume that $|\NNu| < \gamma/9$, as otherwise $|\NNu| \ge \gamma/9 \ge \gamma/(9 \log \delta)$ and we are done.
By partition condition $(P3)$, $|\NNu|  \ge |\GAMm|$.
Hence $\gamma = |\NNu| + |\GAMm| + \gamma'  \le 2|\NNu| + \gamma'$,
yielding $(\gamma-\gamma')/2 \le |\NNu| < \gamma/9$. Hence $2/3 (\gamma/\gamma') \le 6/7$, which gives
$$\log \left(\frac{\gamma}{3}\right) - \log \left(\frac{\gamma'}{2}\right) ~\le~ \log \left(\frac{6}{7}\right) ~\le~ -\frac{2}{9}.$$
It follows that
\begin{equation} \label{eq:finale1}
\frac{\gamma}{\log \delta} \left(\log \delta + \log \left(\frac{\gamma}{3}\right) - \log \left(\frac{\gamma'}{2}\right)\right) ~\le~ \frac{\gamma}{\log \delta} \left(\log \delta -\frac{2}{9}\right).
\end{equation}
Plugging Equation (\ref{eq:finale1}) into Equation (\ref{eq:prefinale1}) gives
$$\gamma' ~\le~ \frac{\gamma}{\log \delta} \left(\log \delta -\frac{2}{9}\right) ~=~ \gamma - \frac{2\gamma}{9 \log \delta} ~\le~ 2|\NNu| + \gamma' - \frac{2\gamma}{9 \log \delta},$$
implying that $|\NNu| \ge \gamma/(9 \log \delta)$.
\QED
\end{proof}
By corollaries \ref{tightlemma}, \ref{lower bound on gamma1 with t and c} and \ref{improving shay's lower bound} we get the following result. Denote  \[M_G(x)= \max \left\{
  \begin{array}{lr}
   \min \{ 1 /(9 \log x),1/20\},\\
   1/(9 \log (2x)),\\
    \max\{(1-1/t) (2.0087/\log(tx))~|~t>1\}
  \end{array}
\right\}.
\]
\begin{corollary} \label{conclusion of the lower bounds}
In $G_S$, there is a subset $S'$ of $S$ with
$|\UN_{S}(S')|\ge \gamma\cdot M_G(\delta)$ .

\end{corollary}

\begin{observation}
%\begin{align*}

$\max\{\min \{ \gamma /(9 \log \delta),\gamma/20\},\gamma /(9 \log (2\delta))\}$  is given by
\begin{eqnarray*}
\begin{cases}
%\hspace{-15pt}
    \gamma/(9 \log (2\delta))       & \quad \text{if } \delta\leq 2^{11/9}\\
    \gamma/20\  & \quad \text{if } 2^{11/9}\leq\delta\leq 2^{20/9}\\
    \gamma /(9 \log \delta)  & \quad \text{otherwise}.
  \end{cases}
 \end{eqnarray*}
% \end{align*}

Moreover, for every $\epsilon>0$, if $\delta$ satisfies  $\epsilon\ln(\delta)-\ln(\ln\delta)-\ln(1+\epsilon)-1\geq 0$, then $\max\{\gamma(1-1/t) (2.0087/\log(t\delta))~|~ t>1\}= \gamma \frac{2.0087}{(1+\epsilon)\log(\delta)}$. In that case, \\$\max\{ \gamma/(9 \log \delta) ,\max\{\gamma(1-1/t) (1/(2(1+c)\log_c(t\delta))~|~t>1\}\}\geq\gamma \frac{2.0087}{(1+\epsilon)\log(\delta)}$ if and only if $\epsilon<17.0783$, i.e., to understand which expression is the maximum, we need to take $\epsilon'=\min\{\epsilon~|~\epsilon\ln(\delta)-\ln(\ln\delta)-\ln(1+\epsilon)-1\geq 0 \}$ and then check if $\epsilon'<17.0783$ or not.
\end{observation}

Let $G=(V,E)$ be an $(\alpha, \beta)$-expander, and for every $S$ in $V$, denote $\gamma_S=|\EN|$. As $G$ is an $(\alpha, \beta)$-expander, $\gamma_S\geq\beta|S|$.
Then, Corollary \ref{conclusion of the lower bounds} yields the following bound on $\beta_w$.
\begin{lemma}
\label{lem:final-M}
Let $G=(V,E)$ be an $(\alpha, \beta)$-expander.
Then,
\begin{description}
\item{(1)} $G$ is an $(\alpha_w,\beta_w)$-wireless expander with $\alpha_w=\alpha$ and $\beta_w \ge \beta \cdot M_G(\bar{\delta}).$
\item{(2)} In the regime $\beta \ge 1$, we have $\delta_S \le \Delta / \beta$, thus $\bar{\delta}\leq \Delta/\beta$, and hence $\beta_w \ge \beta \cdot M_G(\Delta/\beta)$.
\end{description}
\end{lemma}
\begin{proof}
Let $S$ in $V$ s.t. $|S|\leq \alpha n$, and let \\$G_S=(S,\EN, E_S)$ be the corresponding graph. Then, by Corollary \ref{conclusion of the lower bounds},
$|\UN_{S}(S')|\geq  \gamma_S\cdot M_G(\delta_S)\ge \beta|S|\cdot M_G(\delta_S).$ Now, $M_G(x)$ is a decreasing function, and as $\delta_S\geq\bar{\delta}$, we get that $M_G(\delta_S)\geq M_G(\bar{\delta})$ and thus $|\UN_{S}(S')|\ge\beta|S|\cdot M_G(\bar{\delta}).$ Moreover, in the regime $\beta \ge 1$, we have $\delta_S \le \Delta / \beta$, thus $\bar{\delta}\leq \Delta/\beta$, and hence $|\UN_{S}(S')|\ge\beta|S|\cdot M_G(\Delta/\beta)$.
\QED
\end{proof}

The bounds presented in Section \ref{Bounds depending on the average degree} on $\beta_w$ are functions of $\bar{\delta}$ (like the inequality $\beta_w \ge \beta /(9\log (2\bar{\delta}))$ that we proved in Corollary \ref{tightlemma1}). Theses bounds are usually hard to use, since in most cases we cannot give an  evaluation of $\bar{\delta}$. But there are cases in which we can evaluate $\bar{\delta}$, and get a better lower bound for $\beta_w$ than $\beta /(9 \log (2\Delta))$. One such example is the class of bounded arboricity graphs.
}%\APPENDDETERMENSTIC
%%%%%%%%%%%%%%%%%%%%
%\clearpage
%{\small
%\bibliographystyle{abbrv}
%vspace{4pt}

%\bibliographystyle{latex8}
%\bibliography{expander}
%}

%%%%%%%%%%%%%%%%%%%%
%\clearpage
%\appendix

\clearpage
\pagenumbering{roman}
\appendix
\centerline{\LARGE\bf Appendix}
 
%\centerline{\bf Appendix}

%\section{Proof of Lemma 3.1}\label{app:relation}

%\inline Proof of Lemma \ref{lem:uniq1}:

%\inline Proof of Lemma \ref{bad unique expander}:
%\APPENDB

%\section{Some Missing Proofs for Section\ \ref{sec:neg}}

%%%%%%%%%%%%%%%%
%\subsection{Proof of Lemma \ref{coregen}}
%\label{append:corgen}
%\APPENDLEMMACORGEN

%%%%%%%%%%%%%%%%
%\subsection{Missing Proofs for Section\ \ref{plug}} \label{app:helper}

%\inline Proof of Corollary \ref{cl:ordinary}:
%\APPENDORD

%\inline Proof of Claim \ref{cl:helperwire}:
%\APPENDCLHELPW

\section{Deterministic and Constructive Analysis with Improved Bounds}
\label{sec:detbounds}
\APPENDDETERMENSTIC

%\section{Figures}

%\FIGA

%\FIGB

%\clearpage

%\inline Proof of Lemma \ref{constructive lemma}:
%\APPENDE
%
%\inline Proof of Cor. \ref{basiclemma1}:
%\APPENDF
%
%\inline Proof of Lemma \ref{partition using the degrees}:
%\APPENDH
%
%\inline Proof of Cor. \ref{samedegrees}:
%\APPENDI
%
%\inline Proof of Cor. \ref{lower bound on gamma1 with t and c}:
%\APPENDJ
%
%\inline Proof of Cor. \ref{finding t and c}:
%\APPENDK
%
%\inline Proof of Lemma \ref{lem:c and t}:
%\APPENDL
%
%\inline Proof of Lemma \ref{tightlemma}:
%\APPENDM
%
%\inline Proof of Cor. \ref{improving shay's lower bound}:
%\APPENDN
%
%\inline Proof of Lemma \ref{lem:final-M}:
%\APPENDO

\end{document}